\documentclass[journal]{IEEEtran}
\usepackage{amsmath}
\usepackage{amssymb}
\usepackage{upref}
\usepackage{amsfonts}
\usepackage{graphicx}
\usepackage{epic}
\usepackage{eepic}
\usepackage{psfrag}
\usepackage{wasysym}
\usepackage{subfigure}

\newcommand{\abs}[1]{\left|#1\right|}

\newcommand{\field}[1]{\mathbb{#1}}
\newcommand{\Z}{\field{Z}}

\newcommand{\cH}{{\cal H}}
\newcommand{\cA}{{\cal A}}

\newcommand{\cC}{{\cal C}}

\newcommand{\cS}{{\cal S}}

\newcommand{\roundb}[1]{\left( #1 \right)}

\newcommand{\code}{\mathcal{C}}

\newtheorem{theorem}{Theorem}

\newtheorem{lemma}{Lemma}

\begin{document}

\bibliographystyle{IEEEtran}

\title{Error-Correction of Multidimensional Bursts}
\author{Tuvi Etzion,~\IEEEmembership{Fellow,~IEEE} and Eitan Yaakobi,~\IEEEmembership{Student Member,~IEEE}
\thanks{T. Etzion is with the Department of Computer Science,
Technion --- Israel Institute of Technology, Haifa 32000, Israel.
(email: etzion@cs.technion.ac.il).}
\thanks{E. Yaakobi was with the Department of Computer Science,
Technion --- Israel Institute of Technology, Haifa 32000, Israel.
He is now with the department of Electrical and Computer
Engineering, University of California, San Diego, La Jolla, CA
92093 (email: eyaakobi@ucsd.edu). This work is part of his M.Sc.
thesis performed at the Technion.}
\thanks{The material in this paper was presented in part in the 2007 IEEE
International Symposium on Information Theory, Nice, France, June
2007.}
\thanks{This work was supported in part by the United States-Israel
Binational Science Foundation (BSF), Jerusalem, Israel, under
Grant 2006097.} }

\maketitle

\begin{abstract}
In this paper we present several methods and constructions to
generate codes for correction of a multidimensional cluster-error.
The goal is to correct a cluster-error whose shape can be a
box-error, a Lee sphere error, or an error with an arbitrary
shape. Our codes have very low redundancy, close to optimal, and
large range of parameters of arrays and clusters. Our main results
are summarized as follows:
\begin{enumerate}
\item A construction of two-dimensional codes capable to correct a
rectangular-error with considerably more flexible parameters from
previously known constructions. Another advantage of this
construction over previously known constructions is that it is
easily generalized for $D$ dimensions.

\item A novel method based on $D$ colorings of the $D$-dimensional
space for constructing $D$-dimensional codes correcting
$D$-dimensional cluster-error of various shapes. This method is
applied efficiently to correct a $D$-dimensional cluster error
with parameters not covered efficiently by the previous
constructions.

\item A transformation of the $D$-dimensional space into another
$D$-dimensional space in a way that a $D$-dimensional Lee sphere
is transformed into a shape located in a $D$-dimensional box of a
relatively small size. This transformation enables us to use the
previous constructions to correct a $D$-dimensional error whose
shape is a $D$-dimensional Lee sphere.

\item Applying the coloring method to correct more efficiently a
two-dimensional error whose shape is a Lee sphere. The
$D$-dimensional case is also discussed.

\item A construction of one-dimensional codes capable to correct a
burst-error of length $b$ in which the number of erroneous
positions is relatively small compared to $b$. This construction
is generalized for $D$-dimensional codes.

\item Applying the construction for correction of a Lee sphere
error and the construction for correction of a cluster-error with
small number of erroneous positions, to correct a $D$-dimensional
arbitrary cluster-error.
\end{enumerate}
\end{abstract}


\section{Introduction}
\label{sec:introduction}

In current memory devices for advanced storage systems the
information is stored in two or more dimensions. In such systems
errors usually take the form of multidimensional bursts. Usually,
a cluster of errors either will be affected by the position in
which the error event occurred or will be of an arbitrary shape.
But, since an arbitrary cluster-error is hard to correct
efficiently it is common to assume some type of cluster-error (as
any arbitrary cluster is located inside a cluster with a certain
shape, e.g., any two-dimensional cluster is located inside a
rectangle). These types of errors can be of specific shapes like
rectangles or Lee spheres. We will consider these types of errors
as well as arbitrary cluster-errors. The main measure to compute
the efficiency of a cluster-error correcting code is its
redundancy. If we want to design a code which corrects one
multidimensional cluster-error with volume $B$ (of an arbitrary or
a specific shape) then the redundancy of the code $r$ satisfies $r
\geq 2B$. This bound, known as the Reiger bound~\cite{Rei60}, is
attained for binary two-dimensional codes, which correct a
rectangular error, constructed recently by Boyarinov~\cite{Boy06}.
If the volume of the array is $N$ then the redundancy of the code
must also satisfy $r \geq \text{log}_2 N+B-1$ (usually $r- \lceil
\text{log}_2 N \rceil \geq B$). The difference $r- \lceil
\text{log}_2 N \rceil$ will be called the {\it excess redundancy}
of the code~\cite{Abd86,AMT} (even so our definition is slightly
different). Abdel-Ghaffar~\cite{Abd86} constructed a binary
two-dimensional code which corrects a burst with a rectangle shape
for which $r = \lceil \text{log}_2 N \rceil +B$. The code has a
few disadvantages: very limited size, complicated construction,
and there is no obvious generalization for higher dimensions.

Our goal is to design codes which are capable to correct a
cluster-error whose shape is a box, a Lee sphere, or an arbitrary
shape. The method should be able to work on two-dimensional codes
and multidimensional codes, and the parameters of the size of the
codewords and the size of the cluster are as flexible as possible.
There will be a price for our flexibility and our ability to
generalize a two-dimensional construction into multidimensional
construction. This price will be in the excess redundancy. While
the two-dimensional codes which correct a rectangle-error of
Abdel-Ghaffar~\cite{Abd86} have optimal excess redundancy, the
excess redundancy of our codes is only close to optimality.
Moreover, the novel methods enable us to correct a cluster whose
shape is a Lee sphere and an arbitrary cluster, with excess
redundancy close to optimal or very low, depending on the exact
parameters. Previous to our methods the way to correct such a
cluster-error was to use a code which corrects a box-error in
which the cluster-error is located, a method for which the excess
redundancy is far from optimality.

The rest of the paper is organized as follows. In
Section~\ref{sec:known} we briefly survey some of the known
constructions which are essential to understand our results. In
Section~\ref{sec:multi} we present a construction for codes which
correct a multidimensional box-error. The construction is a
generalization and a modification of the construction of
Breitbach, Bossert, Zybalov, and Sidorenko~\cite{BBZS} for
correction of bursts of size $b_1 \times b_2$. Better codes are
constructed when the volume of the $D$-dimensional box-error is an
odd integer. These constructions and the constructions which
follow use auxiliary codes, called component codes, one code for
each dimension. In Section~\ref{sec:coloring} we present a novel
method for correction of a $D$-dimensional cluster. The
construction uses $D$ colorings of the $D$-dimensional space. The
construction of Section~\ref{sec:multi} is a special case of this
construction. The new construction enables us to handle different
burst patterns. In Section~\ref{sec:coloring} we use this method
to handle $D$-dimensional box-error, where the volume of the box
is an even integer. In Section~\ref{sec:Lee} we discuss how to
correct a $D$-dimensional cluster-error whose shape is a Lee
sphere. Two types of constructions are used. The first one uses a
transformation of the $D$-dimensional space into another
$D$-dimensional space in a way that each Lee sphere is transformed
into a shape located inside a reasonably small $D$-dimensional
box, so that we can use the constructions of the previous
sections. The transformation is especially efficient for
two-dimensions. The second construction uses colorings as in
Section~\ref{sec:coloring}. For two-dimensional array the
colorings that we use result in codes with excess redundancy close
to optimality (where our measure for optimality is the lower bound
on the excess redundancy), which improves on the construction
obtained by the two-dimensional transformation. The generalization
for multidimensional Lee sphere errors usually does not make the
same improvement. This is also discussed in Section~\ref{sec:Lee}.
In Section~\ref{sec:limitweight} we show how we handle bursts of
size $b$, where the number of erroneous positions is limited.
First, we present a construction for one-dimensional codes and
afterwards we generalize it into $D$-dimensional codes. In
Section~\ref{sec:arbitrary} we combine the constructions of codes
which are capable to correct Lee sphere error and the construction
capable to correct a burst with a limited number of erroneous
positions for a construction of codes capable to correct arbitrary
bursts. In Section~\ref{sec:parity-check} we describe codes with
the same or slightly better parameters than the parameters of the
codes from the previous sections by using parity-check matrices.
Finally, a conclusion and a list of problems for further research
are given in Section~\ref{sec:conclusion}.

\section{Known Constructions}
\label{sec:known}

Five constructions are important to understand our construction
and their comparison with previous results.

\begin{itemize}
\item
Abdel-Ghaffar, McEliece, Odlyzko, and van Tilborg~\cite{AMOT}
construction of optimum binary cyclic burst-correcting codes.

\item
Abdel-Ghaffar~\cite{Abd88} construction of optimum cyclic
burst-correcting codes over GF($q$).

\item Abdel-Ghaffar construction~\cite{Abd88} of two-dimensional
codes which correct rectangular-error of size $b_1 \times b_2$
with excess redundancy $b_1 b_2$.

\item Breitbach, Bossert, Zybalov, and Sidorenko
construction~\cite{BBZS} of two-dimensional codes for correction
of a ($b_1 \times b_2$)-rectangular-error by using vertical and
horizontal component codes. \item Abdel-Ghaffar, McEliece, and van
Tilborg construction~\cite{AMT} of two-dimensional burst
identification codes, which are used to identify the shape of an
error and together with burst location codes are used for
correction of a two-dimensional cluster.

\end{itemize}

Most of two-dimensional codes known in the literature are designed
to correct a single cluster-error of size $b_1 \times
b_2$~\cite{Abd86,AMT,Boy06,BBZS,Ima73} (only in some recent
papers~\cite{BBV,EtVa,ScEt} it is assumed that the cluster-error
can have an arbitrary shape). Two of these methods are important
in our discussion. Abdel-Ghaffar~\cite{Abd86} gave a construction
of such $n_1 \times n_2$ code with excess redundancy $b_1 b_2$.
One disadvantage of his method is that $n_2$ must be considerably
larger than $n_1$ (with a possible exception when $b_2 \leq 2$,
subject to a list of restrictive conditions), and the existence of
the code depends on series of restricted conditions. The main goal
of his construction was to show that for any given integers $b_1$
and $b_2$ there exists a cyclic $(b_1 \times
b_2)$-cluster-correcting code of some size $n_1 \times n_2$ having
optimal excess redundancy. Therefore, the size of the array was
not a factor in his construction. His construction is a
generalization of the optimum cyclic one-dimensional codes which
correct a single cyclic burst of length $b$~\cite{Abd88,AMOT}.
Over GF($q$) such code has length $n$, redundancy $r$, and it can
correct a single cyclic burst of length $b \geq 1$, where
$n=\frac{q^{r-b+1}-1}{q-1}$. The existence of such codes was
obtained by the following necessary and sufficient conditions.
\begin{theorem}
\label{thm:necessary} If a polynomial $g(x)$ generates an optimum
$b$-burst-correcting code over GF($q$), then it can be factored as
$g(x)=e(x)p(x)$, where $e(x)$ and $p(x)$ satisfy the conditions:
\begin{enumerate}
\item $e(x)$ is a square-free polynomial of degree $b-1$ which is
not divisible by $x$ such that $h_e$ and $m_e$ are relatively
primes to $q-1$, where $h_e$ and $m_e$ are the period of $e(x)$
and the degree of the splitting field of $e(x)$, respectively.

\item $p(x)$ is an irreducible polynomial of degree $m \geq b+1$
and period $\frac{q^m-1}{q-1}$ such that $m$ and $q-1$ are
relatively primes and $m \equiv 0~(\text{mod}~m_e)$.

\end{enumerate}

\end{theorem}
A monic polynomial over GF($q$) which satisfies condition 1) of
Theorem~\ref{thm:necessary} will be called a {\it $b$-polynomial}.
\begin{theorem}
\label{thm:sufficient} Let $e(x)$ be a $b$-polynomial over
GF($q$). Then, for all sufficiently large $m$ relatively prime to
$q-1$ such that $m \equiv 0~(\text{mod}~m_e)$, where $m_e$ is the
degree of the splitting field of $e(x)$, there exists an
irreducible polynomial $p(x)$ of degree $m$ such that $e(x)p(x)$
generates an optimum $b$-burst correcting code of length
$\frac{q^m-1}{q-1}$.
\end{theorem}

\noindent {\bf Remark:} If the polynomial $p(x)$ in
Theorems~\ref{thm:necessary} and~\ref{thm:sufficient} is binary
then $p(x)$ in the Theorems is a primitive polynomial.

The second method is due to Breitbach, Bossert, Zybalov, and
Sidorenko~\cite{BBZS}, who gave three constructions of binary
two-dimensional codes of size $n_1 \times n_2$ which correct a
rectangular-error of size $b_1 \times b_2$. Their goal in
presenting these constructions was not to obtain low excess
redundancy, which is one of the goals in our constructions, but to
present new constructions of codes with relatively large array
size and redundancy close to the Reiger's bound. We will use ideas
from one of the constructions which will be called Construction
BBZS.


A codeword $\{ c_{ij} \}$ of the construction has size $n_1\times
n_2 = 2^{b_2}\times 2^{b_1}$ with $4b_1b_2$ redundancy bits
located in positions $\{(i,j)\ : \ 0\leq i \leq 2b_1-1, n_2-b_2
\leq j \leq n_2-1 \} \cup \{(i,j)\ : \ n_1-b_1 \leq i \leq n_1-1,
0\leq j \leq 2b_2-1 \}$ (see Fig.~\ref{constructionBBZS}). These
bits are set initially to be zeroes. Two temporary {\it component
codes} are being used, a vertical code and an horizontal code (see
Fig.~\ref{constructionBBZS}). We will describe the construction of
the vertical code. We note an earlier construction~\cite{RoSe}
which use similar component codes.

For each row $i=2b_1,\ldots,n_1-1$, $b_2$ parity check bits are
generated. $p_{i \ell}$, $\ell=0,1,\ldots,b_2-1$ is computed as
\begin{equation}\label{eq:parity}
    p_{i \ell} = \sum_{j=\ell,\ell+b_2,\ell+2b_2,\ldots,\\ j<n_2}c_{ij}.
\end{equation}
The parity bits $p_{i \ell}$, $\ell=0,1,\ldots,b_2-1$ generate
afterward a symbol
$\underline{p}_i=(p_{i0},p_{i1},\ldots,p_{ib_2-1})$ from the
extension field GF($2^{b_2}$). The symbols $\underline{p}_i$,
$i=2b_1,\ldots,n_1-1$ are considered as the information symbols of
a Reed-Solomon (RS) code of length $n_1$, dimension $n_1 - 2b_1$,
and minimum distance $d = 2b_1 + 1$. By the encoding procedure of
the RS code we obtain $2b_1$ redundancy symbols $\underline{p}_i$,
$i=0,\ldots,2b_1-1$, and the $2b_1 b_2$ upper right corner
redundancy bits of the array are computed in a way that
(\ref{eq:parity}) will hold for $i=0,1,\ldots,2b_1-1$ and
$\ell=0,1,\ldots,b_2-1$. The encoding procedure of the horizontal
code is done in the same manner, where all the $4b_1 b_2$
redundancy bits of the array are assumed to be zeroes (it is
possible to encode also with their new computed values, but we
want to follow the construction exactly as in~\cite{BBZS}).

In the decoding procedure each row generates $b_2$ parity bits
according to (\ref{eq:parity}) such that a word of length $n_1$
over GF($2^{b_2}$) is received (the redundancy bits of the
horizontal code are assumed to be zeroes). Assuming that the error
occurred in the array can be confined inside a rectangle of size
$b_1\times b_2$. The generated word, of the vertical code, has at
most $b_1$ erroneous symbols, which can be corrected by the
decoding procedure of the RS code. The same process is implemented
for the horizontal code. The positions of the erroneous elements
in the vertical codeword induce the rows in which errors occurred
in the array. The positions of the erroneous elements in the
horizontal codeword induce the columns in which errors occurred in
the array. Hence, we locate the positions of the $b_1 \times b_2$
cluster-error in the array. The shape of the cluster, up to $b_2$
horizontal cyclic shifts, is found by the vertical code. The shape
of the cluster up to $b_1$ vertical cyclic shifts, is found by the
horizontal code. As we know the location of the cluster, we can
use one of the component codes to identify the exact shape of the
error and to correct it.\\
{\bf Remark:} Note that the vertical code cannot find errors
inside the lower left corner redundancy bits. The horizontal code
cannot find errors inside the upper right corner redundancy bits.
But, these facts does not affect the decoding procedure, i.e., the
vertical code is able to know the erroneous rows even if some the
erroneous bits are the lower left redundancy bits.

\begin{figure}[htbp]
\begin{center}
\setlength{\unitlength}{13bp}
\begin{picture}(-1,13)(8,-2.5)
\linethickness{0.2mm} \resizebox{!}{45mm}{
\put(0,0){\line(1,0){16}} \put(0,16){\line(1,0){16}}
\put(0,0){\line(0,1){16}} \put(16,0){\line(0,1){16}}
\linethickness{0.01mm}

\put(17,0){\line(0,1){16}} \put(21,0){\line(0,1){16}}
\put(17,0){\line(1,0){4}} \put(17,16){\line(1,0){4}}
\put(17,10){\line(1,0){4}}

\put(0,-1){\line(1,0){16}} \put(0,-4){\line(1,0){16}}
\put(0,-1){\line(0,-1){3}} \put(16,-1){\line(0,-1){3}}
\put(8,-1){\line(0,-1){3}} 

\put(0,7){\line(1,0){16}} \put(0,8){\line(1,0){16}}
\put(17,7){\line(1,0){4}} \put(17,8){\line(1,0){4}}
\put(-0.4,7.25){$i$}

\put(17.1,7.35){\small{$p_{i0}\cdots p_{ib_2-1}$}}
\put(21.3,7.25){\small{$\underline{p}_i$}}

\put(0,0){\line(1,0){8}} \put(0,3){\line(1,0){8}}
\put(0,0){\line(0,1){3}} \put(8,0){\line(0,1){3}}

\put(2,12){information symbols} \put(2,11){of a codeword}

\put(1.5,2){\small{redundancy symbols}} \put(1.5,1){\small{for
horizontal code}}

\put(12,10){\line(1,0){4}} \put(12,16){\line(1,0){4}}
\put(12,10){\line(0,1){6}} \put(16,10){\line(0,1){6}}

\put(12.5,14){\small{redundancy}} \put(12.5,13){\small{symbols}}
\put(12.5,12){\small{for vertical}} \put(12.5,11){\small{code}}

\put(17.2,14){\small{redundancy}} \put(17.2,13){\small{symbols}}
\put(17.2,12){\small{of RS code}}

\put(17.2,5){\small{information}} \put(17.2,4){\small{symbols}}
\put(17.2,3){\small{of RS code}}

\put(1.5,-2){\small{redundancy symbols}} \put(1.5,-3){\small{of RS
code}}

\put(9,-2){\small{information symbols}} \put(9,-3){\small{of RS
code}}

}

\end{picture}

\caption{Construction BBZS}%

\label{constructionBBZS}

\end{center}
\end{figure}
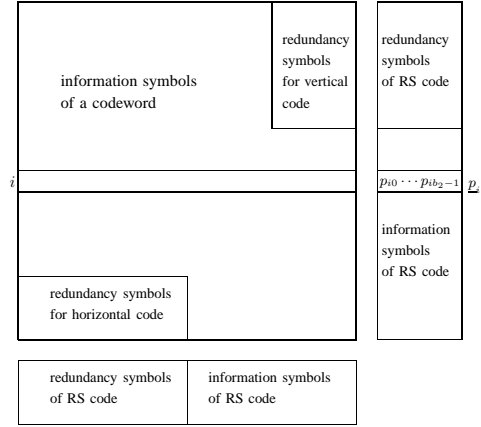

There is no obvious generalization to the construction of
Abdel-Ghaffar~\cite{Abd86} for multidimensional codes, while
immediate generalizations of Construction BBZS cannot support good
redundancy or excess redundancy. One simple way to generalize this
construction is to use the optimum burst-correcting codes
of~\cite{Abd88,AMOT} instead of the RS codes. The vertical
component code over GF($2^{b_2}$) has length $\frac{(2^{b_2})^{r_1
-b_1+1}-1}{2^{b_2}-1}$, redundancy $r_1$, and it can correct a
burst of length $b_1$. The horizontal component code over
GF($2^{b_1}$) has length $\frac{(2^{b_1})^{r_2
-b_21+1}-1}{2^{b_1}-1}$, redundancy $r_2$, and it can correct a
burst of length $b_2$. Instead of $4b_1 b_2$ redundancy bits we
will use $r_1 b_2$ redundancy bits for the vertical code and $b_1
r_2$ redundancy bits for the horizontal code. The excess
redundancy of this construction is $2b_1 b_2-1$ and the excess
redundancy of its generalization for $D$ dimensions is $DB-1$,
where $B$ is the volume of the $D$-dimensional box error.

Further improvements of this construction are presented in the
next section. Henceforth we assume that if a $D$-dimensional code
is discussed then $D$ is a constant. Furthermore, we assume $b_i
>1$ for $1 \leq i \leq D$; this assumption can be made since if
for some $j$, $b_j=1$ then the cluster can be corrected as a
$(D-1)$-dimensional cluster in a $D$-dimensional array.

\section{Construction for Multidimensional Arrays}
\label{sec:multi}

In this section we present our first idea for construction of
multidimensional codes capable to correct a box-error of size $b_1
\times b_2 \times \cdots \times b_D$. First, we will present the
two-dimensional version of the construction. We combine
Construction BBZS with the constructions of Abdel-Ghaffar et
al.~\cite{AMOT} and Abdel-Ghaffar~\cite{Abd88} to obtain codes
with variety of parameters. The redundancy of the construction is
kept relatively small as our horizontal code will find only the
location of the error and not its shape. This idea is the first
key of all our constructions. The second idea to reduce the
redundancy is to use a binary horizontal code instead of a code
over GF($2^{b_1}$). Finally, the structure of the construction
makes it possible to generalize it to any dimension. The
generalization is relatively quite simple, with low redundancy,
and can be applied on a large range of parameters. One
disadvantage is that the construction is defined for a box-error
whose volume is an odd integer. To apply the construction on a
box-error whose volume is an even integer, we have to increase the
box-error artificially such that its volume will be an odd integer
and the real box error will be located inside the artificial box
error. This will cost us extra unnecessary redundancy. In the next
section we will solve this problem by giving a novel construction
for correction of a box-error whose volume is an even integer.

\subsection{Two-dimensional codes}

The vertical component code of  Construction BBZS finds the rows
in which the burst occurred and the shape of the cluster up to a
cyclic permutation of the columns. Hence, the work done by the
horizontal code to find the shape of the cluster is redundant.
Therefore, we want to find an horizontal component code that will
determine only the first column of the cluster. More explicitly, a
burst $e=(e_0,e_1,\ldots,e_{b_2-1})$, where $e_i\in$
GF($2^{b_1}$), for $0\leq i \leq b_2-1$ found by the vertical
code, can start at any column $0\leq i \leq n_2-b_2$ (See
Fig.~\ref{fig:cluster-error}). However, if the first column of the
cluster is $i$, then the cluster occurred is
$e'=(e_{i_0},e_{i_0+1},\ldots,e_{i_0+b_2-1})$ where $i_0\equiv
i(\text{mod}~b_2)$, and indices are taken modulo $b_2$.

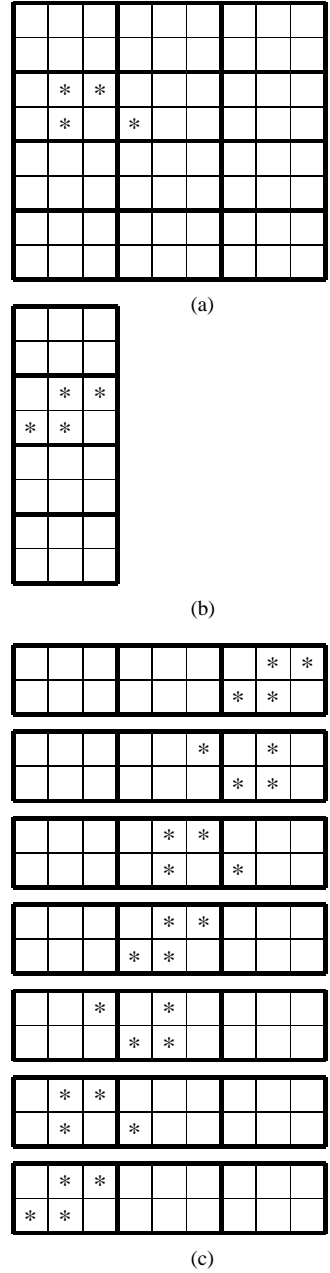
\begin{figure}
\begin{center}
\subfigure[]{ \setlength{\unitlength}{0.46cm}
\begin{picture}(-5.5,8)(8,0)
\linethickness{0.075mm} \multiput(0,0)(1,0){10} {\line(0,1){8}}
\multiput(0,0)(0,1){9} {\line(1,0){9}}%
\linethickness{0.5mm} \multiput(0,0)(0,2){5} {\line(1,0){9}}
\multiput(0,0)(3,0){4} {\line(0,1){8}} %
\put(1.3,5.15){*} \put(2.3,5.15){*} \put(1.3,4.15){*}
\put(3.3,4.15){*}
\end{picture}
\label{error1} }

\subfigure[]{ \setlength{\unitlength}{0.46cm}
\begin{picture}(-5.5,7)(8,0) \linethickness{0.075mm}
\multiput(0,0)(1,0){4} {\line(0,1){8}}
\multiput(0,0)(0,1){9} {\line(1,0){3}}%
\linethickness{0.5mm} \multiput(0,0)(0,2){5} {\line(1,0){3}}
\multiput(0,0)(3,0){2} {\line(0,1){8}}%
\put(1.3,5.15){*} \put(2.3,5.15){*} \put(1.3,4.15){*}
\put(0.3,4.15){*}
\end{picture}
\label{error2}}

\subfigure[]{ \setlength{\unitlength}{0.46cm}
\begin{picture}(-5.5,17)(8,0)
\linethickness{0.075mm} \multiput(0,0)(1,0){10} {\line(0,1){2}}
\multiput(0,0)(0,1){3} {\line(1,0){9}}%
\linethickness{0.5mm} \multiput(0,0)(0,2){2} {\line(1,0){9}}
\multiput(0,0)(3,0){4} {\line(0,1){2}}%
\put(1.3,1.15){*} \put(2.3,1.15){*} \put(1.3,0.15){*}
\put(0.3,0.15){*} %
\linethickness{0.075mm} \multiput(0,2.5)(1,0){10} {\line(0,1){2}}
\multiput(0,2.5)(0,1){3} {\line(1,0){9}}%
\linethickness{0.5mm} \multiput(0,2.5)(0,2){2} {\line(1,0){9}}
\multiput(0,2.5)(3,0){4} {\line(0,1){2}}%
\put(3.3,2.65){*} \put(2.3,3.65){*} \put(1.3,3.65){*}
\put(1.3,2.65){*}%
\linethickness{0.075mm} \multiput(0,5)(1,0){10} {\line(0,1){2}}
\multiput(0,5)(0,1){3} {\line(1,0){9}}%
\linethickness{0.5mm} \multiput(0,5)(0,2){2} {\line(1,0){9}}
\multiput(0,5)(3,0){4} {\line(0,1){2}}%
\put(3.3,5.15){*} \put(2.3,6.15){*} \put(4.3,6.15){*}
\put(4.3,5.15){*}%
\linethickness{0.075mm} \multiput(0,7.5)(1,0){10}  {\line(0,1){2}}
\multiput(0,7.5)(0,1){3} {\line(1,0){9}}%
\linethickness{0.5mm} \multiput(0,7.5)(0,2){2} {\line(1,0){9}}
\multiput(0,7.5)(3,0){4} {\line(0,1){2}}%
\put(4.3,8.65){*} \put(5.3,8.65){*} \put(4.3,7.65){*}
\put(3.3,7.65){*}%
\linethickness{0.075mm} \multiput(0,10)(1,0){10} {\line(0,1){2}}
\multiput(0,10)(0,1){3} {\line(1,0){9}}%
\linethickness{0.5mm} \multiput(0,10)(0,2){2} {\line(1,0){9}}
\multiput(0,10)(3,0){4} {\line(0,1){2}}%
\put(6.3,10.15){*} \put(5.3,11.15){*} \put(4.3,11.15){*}
\put(4.3,10.15){*}%
\linethickness{0.075mm} \multiput(0,12.5)(1,0){10} {\line(0,1){2}}
\multiput(0,12.5)(0,1){3} {\line(1,0){9}}%
\linethickness{0.5mm} \multiput(0,12.5)(0,2){2} {\line(1,0){9}}
\multiput(0,12.5)(3,0){4} {\line(0,1){2}}%
\put(6.3,12.65){*} \put(5.3,13.65){*} \put(7.3,13.65){*}
\put(7.3,12.65){*}%
\linethickness{0.075mm} \multiput(0,15)(1,0){10} {\line(0,1){2}}
\multiput(0,15)(0,1){3} {\line(1,0){9}}%
\linethickness{0.5mm} \multiput(0,15)(0,2){2} {\line(1,0){9}}
\multiput(0,15)(3,0){4} {\line(0,1){2}}%
\put(7.3,16.15){*} \put(8.3,16.15){*} \put(7.3,15.15){*}
\put(6.3,15.15){*}
\end{picture}
\label{error3} }

\caption{Figure \ref{error1} describes a cluster-error in a
($2\times 3$)-burst-correcting code. The cluster found by the
vertical code is demonstrated in figure \ref{error2}, and figure
\ref{error3} shows the possible clusters considered by the
horizontal code.}

\label{fig:cluster-error}
\end{center}
\end{figure}

Our new construction, in which the horizontal code only locates
the first column of the cluster, is based on two lemmas. The first
one is proved here only for the binary case.

\begin{lemma}
\label{lem:b-polynomial_gf_2} If $e_2(x)= 1+x+x^2+\cdots +
x^{b-1}$ and $b$ is an odd integer then $e_2(x)$ is a
$b$-polynomial over GF(2).
\end{lemma}

\begin{proof}
Clearly, $e_2(x)$ is not divisible by $x$. The derivative of
$x^b-1$ over GF(2) is $x^{b-1}$, and since
g.c.d.$(x^b-1,x^{b-1})=1$, it follows that $x^b-1$ is a
square-free polynomial and hence $e_2(x)$ is also square-free.
Therefore, by Theorem~\ref{thm:necessary}, $e_2(x)$ is a
$b$-polynomial over GF(2).
\end{proof}

There is an alternative more general version of
Lemma~\ref{lem:b-polynomial_gf_2}.

\begin{lemma}
\label{lem:b-polynomial_gf_q} Let $e_2(x)$ be the polynomial
$e_2(x)= 1+x+x^2+\cdots + x^{b_2-1}$ over GF($2^{b_1}$), where
$b_1$ and $b_2$ are positive integers. Assume that the following
conditions hold:
\begin{enumerate}
    \item $\gcd(b_2,2)=1$,
    \item $\gcd(b_2,2^{b_1}-1)=1$,
    \item $\gcd(\phi (b_2),2^{b_1}-1)=1$.
\end{enumerate}
Then, $e_2(x)$ is a $b_2$-polynomial over GF($2^{b_1}$).
\end{lemma}

We omit the proof of the lemma (see~\cite{Yaa07}) as the
construction which uses the lemma, and is described in this
subsection, has inferior redundancy than the one described in the
next subsection and uses Lemma~\ref{lem:b-polynomial_gf_2}. We
will compare these redundancies in the sequel. The code over
GF($2^{b_1})$ is described since it is a bridging step to
understand the one over GF(2).

By Theorem~\ref{thm:sufficient}, for the $b_2$-polynomial $e_2(x)=
1+x+x^2+\cdots + x^{b_2-1}$, over GF($2^{b_1}$), there exists an
irreducible polynomial $p_2(x)$ of degree $m_2=r_2-b_2+1$ such
that $e_2(x)p_2(x)$ generates an optimum $b_2$-burst-correcting
code $\cC^*$ of length $n_2=\frac{(2^{b_1})^{m_2}-1}{2^{b_1}-1}$
and redundancy $r_2$. Let
$$\cC^*=\{f(x)\in
\text{GF}(2^{b_1})[x]  :  e_2(x)p_2(x)|f(x), \ \deg f(x) < n_2\},
$$
and let
$$
\code_2=\{f(x)\in \text{GF}(2^{b_1})[x] : p_2(x)|f(x), \ \deg f(x)
< n_2\}.
$$
The code $\code_2$ is also of length $n_2$ and has $m_2$
redundancy symbols over GF($2^{b_1}$). We will show now that
$\code_2$ can serve as the horizontal component code, i.e., given
that the cluster occurred up to a cyclic permutation is
$e=(e_0,e_1,\ldots,e_{b_2-1})$, where $e_i\in \text{GF}(2^{b_1})$,
for $0\leq i \leq b_2-1$, it will be possible to determine the
first column $i$, $0\leq i \leq n_2-b_2$, of the cluster with
$\cC_2$.

\begin{lemma}
\label{lem:locator} Let $e=(e_0,e_1,\ldots,e_{b_2-1})$ be a given
cluster, up to a cyclic permutation, which occurred in a
transmitted codeword and found by the vertical component code.
Then, the horizontal component code $\code_2$ can determine the
first column of the given burst.
\end{lemma}

\begin{proof}
We have to prove that if the burst $e$, or a cyclic shift of $e$,
occurred in two different codewords $f_1(x),f_2(x)$ then two
different words will be generated. Since $\code_2$ is a linear
code it is sufficient to prove that there is no codeword which is
equal to the difference of two clusters which are cyclic shifts of
the cluster $e$. Assume the first column of the cluster $e$ is
$i$, i.e., the cluster is
$e'=(e_{i_0},e_{i_0+1},\ldots,e_{i_0+b_2-1})$ where $i_0\equiv
i(\text{mod}~b_2)$ and indices are taken modulo $b_2$. The
polynomial representing this cluster is
\begin{eqnarray}
h_i(x) & = & x^{\ell_i b_2}(e_{i_0}x^{i_0} + e_{i_0+1}x^{i_0+1} + \cdots + e_{b_2-1}x^{b_2-1} + \nonumber\\
& & e_0x^{b_2} + e_1x^{b_2+1} + \cdots + e_{i_0-1}x^{b_2+i_0-1})
\nonumber,
\end{eqnarray}
where $\ell_i=\lfloor \frac{i}{b_2} \rfloor$. We can write
$h_i(x)$ as
\begin{eqnarray}
\lefteqn{h_i(x)= x^{\ell_i b_2}\big(e_0 + e_1x+\cdots+e_{b_2-1}x^{b_2-1} + e_0(x^{b_2}-1) + }\nonumber\\
& & e_1(x^{b_2+1}-x) + \cdots + e_{i_0-1}(x^{b_2+i_0-1}-x^{i_0-1})\big) = \nonumber \\
& & x^{\ell_i b_2}\big(e_0 + e_1x + \cdots + e_{b_2-1}x^{b_2-1} + \nonumber \\
& & (x^{b_2}-1)(e_0 + e_1x + \cdots + e_{i_0-1}x^{i_0-1})\big) = \nonumber \\
& & x^{\ell_i b_2}\roundb{e_0 + e_1x + \cdots + e_{b_2-1}x^{b_2-1}} \nonumber \\
& & + x^{\ell_i b_2}(x^{b_2}-1)\roundb{e_0 + e_1x + \cdots +
e_{i_0-1}x^{i_0-1}}. \nonumber
\end{eqnarray}
Assume the contrary, that the difference between two clusters
which are cyclic shifts of the cluster $e$ is a codeword. Assume
that these two clusters start at columns $i$ and $j$. Hence, the
polynomial $h_i(x)-h_j(x)$ is the codeword
\begin{eqnarray}
\lefteqn{h_i(x)-h_j(x)=x^{\ell_i b_2}\roundb{e_0 + e_1x + \cdots +
e_{b_2-1}x^{b_2-1}} + }\nonumber \\
& & x^{\ell_i b_2}(x^{b_2}-1)\roundb{e_0 + e_1x + \cdots + e_{i_0-1}x^{i_0-1}} -\nonumber \\
& & x^{\ell_j b_2}\roundb{e_0 + e_1x + \cdots +
e_{b_2-1}x^{b_2-1}} - \nonumber \\
& & x^{\ell_j b_2}(x^{b_2}-1)\roundb{e_0 + e_1x + \cdots +
e_{j_0-1}x^{j_0-1}}. \nonumber
\end{eqnarray}
$h_i(x)-h_j(x)$ can be written as
\begin{eqnarray}
\lefteqn{h_i(x)-h_j(x) = (x^{\ell_i b_2} - x^{\ell_j b_2})}\nonumber \\
& & \big(e_0 + e_1x + \cdots +
e_{b_2-1}x^{b_2-1}\big) + \nonumber \\
& & (x^{b_2}-1)\big(x^{\ell_i b_2}\roundb{e_0 + e_1x + \cdots +e_{i_0-1}x^{i_0-1}} \nonumber \\
& & - x^{l_j b_2}\roundb{e_0 + e_1x + \cdots + e_{j_0-1}x^{j_0-1}}\big) = \nonumber \\
& & x^{\ell_j b_2}(x^{(\ell_i-\ell_j)b_2}-1) \roundb{e_0 + e_1x +
\cdots + e_{b_2-1}x^{b_2-1}} + \nonumber \\
& & (x^{b_2}-1)\big(x^{\ell_i b_2}\roundb{e_0 + e_1x + \cdots +
e_{i_0-1}x^{i_0-1}} - \nonumber \\
& & x^{\ell_j b_2}\roundb{e_0 + e_1x + \cdots +
e_{j_0-1}x^{j_0-1}}\big). \nonumber
\end{eqnarray}
This last presentation of $h_i(x)-h_j(x)$ implies that it is
dividable by $x^{b_2}-1$ and hence also by
$e_2(x)=\frac{x^{b_2}-1}{x-1}$. Since $h_i(x)-h_j(x)$ is a
codeword in $\code_2$ it follows that $p_2(x)|h_i(x)-h_j(x)$.
Since also $e_2(x)|h_i(x)-h_j(x)$, $p_2(x)$ is an irreducible
polynomial, and its degree is greater than $b_2$ it follows that
$e_2(x)p_2(x)|h_i(x)-h_j(x)$. Therefore, $h_i(x)-h_j(x)$ is also a
codeword of $\code^*$, a contradiction since $\code^*$ can correct
any burst of length $b_2$.
\end{proof}

A code $\cC$ that can find the first column of a burst
$e=(e_0,e_1, \ldots , e_{b-1})$ given up to a cyclic shift will be
called a {\it $b$-burst-locator code}. Thus, by
Lemma~\ref{lem:locator}, $\cC_2$ is a $b_2$-burst-locator code.
Based on the constructions of~\cite{AMOT,BBZS} and
Lemma~\ref{lem:locator} we can construct an $n_1 \times n_2$
two-dimensional $(b_1\times b_2)$-cluster-correcting code with
small excess redundancy.

Let $\cC_1$ be an optimum $b_1$-burst-correcting code, over
GF($2^{b_2}$), of length
$n_1=\frac{(2^{b_2})^{r_1-b_1+1}-1}{2^{b_2}-1}$ and redundancy
$r_1$. Let $\cC_2$ be a $b_2$-burst-locator code, over
GF($2^{b_1}$), of length $n_2=\frac{(2^{b_1})^{m_2}-1}{2^{b_1}-1}$
and redundancy $m_2=r_2 -b_2 +1$. We can give a construction in
which each codeword of size $n_1 \times n_2$ has $r_1b_2 + m_2 b_1
+1 = \lceil \text{log}_2 (n_1n_2) \rceil + b_1b_2 + b_1 $
redundancy bits. The redundancy bits are partitioned into three
subsets:
\begin{itemize}
\item $r_1 b_2$ redundancy bits are located in the upper right
corner of the array and are computed from the complete vertical
codeword as done in Construction BBZS. \item $m_2 b_1$ redundancy
bits which are computed from the complete horizontal codeword as
done in Construction BBZS. These redundancy bits are spread in the
array in a way that they will fulfill the following requirement.
If a redundancy bit is erroneous it will be possible to determine
its row (note, that if the vertical code finds only one row where
errors occurred, there are $b_1$ different sets of $b_1$ rows in
which the burst occurred). Hence, in $2b_1 -1$ consecutive rows
there can be at most one redundancy bit. This requirement implies
also that in any ($b_1 \times b_2$)-cluster we have at most one
redundancy bit. \item One redundancy bit which is a parity of all
redundancy bits of the second subset. Its role is to determine
whether this bit or a redundancy bit from the second subset is
erroneous.
\end{itemize}

The encoding is done similarly to Construction BBZS with two
exceptions. When we compute the elements of the horizontal
component code the $r_1 b_2$ redundancy bits of the first subset
are not assumed to be zeroes as in Construction BBZS, but have the
values which were computed by the previous steps of the encoding
procedure. The second exception is the extra computation of the
redundancy bit of the third subset, which is taken as an even
parity bit of all the redundancy bits of the second subset.

The decoding is done similarly to Construction BBZS with the
following exceptions. In Construction BBZS, if redundancy bits are
erroneous then they will be recovered by the corresponding
component code. The redundancy bits in the right upper corner are
recovered by the vertical component code. They will be recovered
by this code also in our construction. The redundancy bits in the
left lower corner, of Construction BBZS, are recovered by the
horizontal component code in Construction BBZS. Since, in our
construction we don't use a burst-correcting code as an horizontal
component code we cannot use this code to recover the related
redundancy bits of the second subset. Each $b_1 \times b_2$
sub-codeword can contain at most one redundancy bit from the
second or third subset. By summing all these redundancy bits we
will know if one of them is erroneous. Also, these redundancy bits
are spaced in a way that if we know in which rows errors occurred
then we will know which one of these bits is in error. Once we
will find this erroneous bit we will know the shape of the burst
and the horizontal burst-locator code will find the column in
which the cluster started. If only a redundancy bit from the
second or third subset is in error then the vertical code will not
find erroneous bits. In this case the sum of these bits is odd. If
a bit from the second subset is in error then the horizontal
burst-locator code will correct this error since this code is
generated by a primitive polynomial and hence it can also correct
a single error. Otherwise, the horizontal code will not find an
error, which implies that the redundancy bit of the third subset
is erroneous.

\subsection{Binary burst-locator code}
\label{subsec:binary}

The redundancy of the construction is improved if we use as the
horizontal burst-locator code a binary $(b_1b_2)$-burst-locator
code $\cC$ of length $2^m-1$, where $m=r-b_1b_2+1$, $b_1 b_2$ odd
and $e_2 (x)= 1+x+ \cdots + x^{b_1b_2-1}$. This is done simply by
taking the $b_1$ parity symbols which are computed for each column
as $b_1$ consecutive symbols in $\cC_2$ instead of an element in
GF($2^{b_1}$).

Each codeword $\{ c_{ij} \}$ of size $n_1 \times n_2$ has $r_1b_2
+ m +1$ redundancy bits in the following positions:

\begin{itemize}
\item $\{ (i,j) \ ~:~ \ 0\leq i\leq r_1-1 , n_2-b_2\leq j\leq
n_2-1\}$. These redundancy bits are computed from the vertical
component code.

\item $\{ (i(2b_1-1)+j(2b_1-1)b_1,j) ~:~ \ 0\leq i\leq b_1-1 ,~ 0
\leq j ,~ (j+1)b_1 +i+1 \leq m \}$. These redundancy bits are
computed from the horizontal component code.

\item $\{ (n_1-1,n_2-1)  \}$. This redundancy bit is an even
parity bits for the redundancy bits of the second subset.

\end{itemize}

\noindent {\bf Encoding:}

All the redundancy bits in a codeword are set initially to be
zeroes. The vertical component code and the first set of
redundancy bits are computed as in Construction BBZS, i.e., for
each row $i=r_1,\ldots,n_1-1$, $b_2$ parity check bits are
generated. $p_{i \ell}$, $\ell=0,1,\ldots,b_2-1$ is computed as
\begin{equation}\label{eq:parity1}
    p_{i \ell} = \sum_{j=\ell,\ell+b_2,\ell +2b_2,\ldots,\\ j<n_2}c_{ij}.
\end{equation}
The parity bits $p_{i \ell}$, $\ell=0,1,\ldots,b_2-1$ generate
afterward a symbol
$\underline{p}_i=(p_{i0},p_{i1},\ldots,p_{ib_2-1})$ from the
extension field GF($2^{b_2}$). The symbols $\underline{p}_i$,
$i=r_1,\ldots,n_1-1$ are considered as the information symbols of
the code $\cC_1$. By the encoding procedure for $\cC_1$ we obtain
$r_1$ redundancy symbols $\underline{p}_i$, $i=0,\ldots,r_1-1$,
and the $r_1 b_2$ upper right corner redundancy bits of the array
are computed in a way that (\ref{eq:parity1}) holds for
$i=0,1,\ldots,r_1-1$ and $\ell=0,1,\ldots,b_2-1$. We now turn to
the encoding procedure of the horizontal component code. During
this process the $r_1 b_2$ redundancy bits of the first subset
will have the values which were just computed (as said before this
is different from Construction BBZS, in which they were assumed to
be zeroes). The second subset of redundancy bits spans over
$\lceil \frac{m}{b_1} \rceil$ consecutive columns. In each column,
with a possible exception of the last one, there are $b_1$
redundancy bits in $b_1$ positions (rows) which cover all the
$b_1$ distinct residues modulo $b_1$. We compute $n_2 - m$
information symbols of the horizontal component codeword as in
Construction BBZS. The remaining $m$ symbols are the redundancy
symbols of the horizontal $(b_1b_2)$-burst-locator code and they
are computed from the $n_2 - m$ information symbols. The only
redundancy bit of the third subset is the binary sum of the
computed redundancy bits from the second subset.

\noindent {\bf Decoding:}

The decoding is done similarly to the one of Construction BBZS.
First, each row generates $b_2$ parity bits by using
(\ref{eq:parity1}) (this includes also the first $r_1$ rows, but
the $m+1$ redundancy bits of the second and the third subsets are
assumed to be zeroes). These $b_2$ bits are considered as a symbol
in GF($2^{b_2}$) and hence a word of length $n_1$ over
GF($2^{b_2}$) is generated. Now, we use the decoding procedure of
the vertical $b_1$-burst-correcting code to correct a burst of
length $b_1$. If such a burst occurred it "almost" determines the
rows in which errors occurred and also the shape of the cluster up
to horizontal cyclic permutation. We say "almost" since the
vertical code does not find erroneous redundancy bits from the
second and the third subsets. These bits are spaced in a way that
at most one such bit is in error. We sum these $m+1$ bits and if
the result is not zero then one of these bits is erroneous. If
this is the case then from rows of the cluster-error, discovered
by the vertical code, we will know the exact row of this bit. If
the vertical code didn't find any burst and a redundancy bit from
either the second subset or the third subset is erroneous then we
have exactly a single error in the array. Now, since the
$(b_1b_2)$-burst-locator code is also a single-error-correcting
code (binary Hamming code) this single error can be corrected. If
there are more errors in the cluster then we continue by either
correcting the redundancy bit of the last two subsets (and the
corresponding bit in the horizontal codeword) or adding this
erroneous redundancy bit to the shape of the burst. In either
cases the horizontal $(b_1b_2)$-burst-locator code will discover
the first column in which the cluster occurred, and hence the
pattern discovered by the vertical component code enables us
to correct the errors.\\
{\bf Remark 1:} The parity bits of the second set can be chosen in
other ways as long as they form a set of redundancy symbols for
the burst-locator code, e.g., they don't have to be in consecutive
columns. Such choices can result in other array sizes.\\
{\bf Remark 2:} A natural question is to ask why not to use also a
binary vertical component code? The answer is that we can. The
main advantage will be that we will have more flexibility in the
parameters of our two-dimensional array. The disadvantage is that
the excess redundancy will be increased by $\lceil \text{log}_2
b_2 \rceil$.

The consequence of this construction is the following theorem (the
computational part of the proof will be given in the next
subsection).
\begin{theorem}
The given construction produces a $(b_1 \times b_2
)$-cluster-correcting code of size $n_1 \times n_2 =
{\frac{(2^{b_2})^{r_1-b_1+1}-1}{2^{b_2}-1}} \times {\lfloor
\frac{2^{r-b_1b_2+1}-1}{b_1} \rfloor}$ with redundancy
$r_1b_2+r-b_1b_2+2 \leq \lceil \text{log}_2 (n_1n_2) \rceil +
b_1b_2 + \lceil \text{log}_2 b_1
\rceil$.\\
The construction can be applied whenever $b_1b_2$ is odd integer,
and there exists

\begin{itemize}
\item An optimum $b_1$-burst-correcting code, over GF($2^{b_2}$),
with redundancy $r_1$ and length
$\frac{(2^{b_2})^{r_1-b_1+1}-1}{2^{b_2}-1}$.

\item A binary $(b_1b_2)$-burst-locator code with redundancy $r$
and length $2^{r-b_1b_2+1}-1$.
\end{itemize}
\end{theorem}

\subsection{Multidimensional arrays}

As said earlier, one of the advantages of our construction is that
it can be generalized in relatively simple way to $D$ dimensions,
while the excess redundancy remains relatively small. Assume we
want to construct a $D$-dimensional code of size $n_1 \times n_2
\times \cdots \times n_D$ which corrects a box-error of size $b_1
\times b_2 \times \cdots \times b_D$. Let $B= \prod_{i=1}^D b_i$,
where $B$ is an odd integer. For the first dimension we use a
component code of length $n_1$ over GF$(2^{\frac{B}{b_1}})$ which
corrects a burst error of size $b_1$. In each of the other $D-1$
dimensions we use a burst-locator code which locates the position
of the burst and its cyclic permutation in the corresponding
direction. In dimension $i$, $2 \leq i \leq D$, we use a
$b_i$-burst-locator code of length $n_i$ over
GF$(2^{\frac{B}{b_i}})$. The code of the first dimension finds the
position of the error in the first dimension and the shape of the
error, with a possible cyclic shift in each of the other $D-1$
dimensions. The code in dimension $i$, $2 \leq i \leq D$, finds
the location of the position where the cluster starts in dimension
$i$. After each code discovers the position where the cluster
starts in its dimension (note, that this can be done in parallel),
we have the corresponding cyclic shift in each dimension of the
box-error found by the first code. Hence, we can now form the
actual burst which occurred and correct it. As before, we can use
in dimension $i$, $2 \leq i \leq D$, a binary $B$-burst-locator
code. For the first dimension we can choose consecutive redundancy
bits as the redundancy bits of the first subset. For each other
dimension we will have to choose positions for the redundancy
bits, which will fulfill the requirements for the redundancy bits
of the second subset (only one redundancy bit will be needed as a
parity bit for all the redundancy bits of this form for the
burst-locator codes of all dimensions).

\begin{theorem}
\label{thm:redundancyD} Assume $b_1,b_2,\ldots,b_D$ are odd
integers, $B=\prod_{i=1}^D b_i$ and the following codes exist:

\begin{itemize}
\item
An optimum $b_1$-burst-correcting code, over
GF($2^{\frac{B}{b_1}}$), with redundancy $r_1$ and length
$\frac{(2^{\frac{B}{b_1}})^{r_1-b_1+1}-1}{2^{\frac{B}{b_1}}-1}$.

\item For each $i$, $2 \leq i \leq D$, a binary $B$-burst-locator
code with redundancy $r_i$ and length $2^{m_i}-1=2^{r_i-B+1}-1$.
\end{itemize}
Then, there exists a $(b_1\times b_2\times \cdots \times
b_D)$-burst correcting code of size $n_1\times n_2 \times \cdots
\times n_D =
\frac{(2^{\frac{B}{b_1}})^{r_1-b_1+1}-1}{2^{\frac{B}{b_1}}-1}
\times \left\lfloor\frac{2^{m_2}-1}{b_2}\right\rfloor \times
\cdots \times \left\lfloor\frac{2^{m_D}-1}{b_D}\right\rfloor$ and
redundancy $r_1\frac{B}{b_1} + \sum_{j=2}^Dm_j+1\leq
\left\lceil\log_2(n_1\cdots n_D)\right\rceil + B +
\left\lceil\log_2(b_1B^{D-2})\right\rceil +1$.
\end{theorem}
\begin{proof}
The existence of the code is implied by the proceeding description
and we only have to compute its redundancy. For each $j$, $2\leq j
\leq D$,
$n_j=\left\lfloor\frac{2^{m_j}-1}{\frac{B}{b_j}}\right\rfloor\geq
\frac{2^{m_j}-\frac{B}{b_j}}{\frac{B}{b_j}}$. Therefore,
$n_1n_2\cdots n_D \geq
\frac{(2^{\frac{B}{b_1}})^{r_1-b_1+1}}{2^{\frac{B}{b_1}}}
\prod_{j=2}^D\frac{2^{m_j}-\frac{B}{b_j}}{\frac{B}{b_j}}$. Now,
taking into account that for each $j$, $2 \leq j \leq D$, $m_j >
B$, and w.l.o.g. we can assume that for each $i$, $1 \leq i \leq
D$, $b_i > 1$, we have that for each $j$, $2 \leq j \leq D$,
$2^{m_j}b_j > 2(D-1)B$. It follows that $\log_2 (n_1n_2\cdots n_D)
\geq \frac{Br_1}{b_1} -B + (\sum_{j=2}^D m_j -1) -\sum_{i=2}^D
\log_2 \frac{B}{b_i}=\frac{Br_1}{b_1} -B + \sum_{j=2}^D m_j -
\log_2 (B^{D-2}b_1) -1$. Hence we have that the redundancy of the
code is $r_1\frac{B}{b_1} + \sum_{j=2}^Dm_j+1\leq
\left\lceil\log_2(n_1\cdots n_D)\right\rceil + B +
\left\lceil\log_2(b_1B^{D-2})\right\rceil +1$.
\end{proof}

When $B$ is even we have to modify our method in order to obtain
similar results. The modifications include binary component codes
in all dimensions. Each one of the $D-1$ burst-locator codes
locates the position of a cyclic burst of size $B+1$. This
modification is described in the next section.

\section{Colorings for Error-Correction}
\label{sec:coloring}

The constructions presented in Section~\ref{sec:multi} are best
applied when the volume of the box error is an odd integer. The
reason is that Lemma~\ref{lem:b-polynomial_gf_2} is true only when
$b$ is an odd integer. Hence, if the volume of the box-error is an
even integer then the construction of Section~\ref{sec:multi} has
to be used in a slightly different way. We have to apply the
construction for correcting a box-error which has odd volume and
contains the "real" box-error. The price will be an increase in
the excess redundancy. In this section we offer a novel method
which will be useful to correct a box-error with even volume and
also for correcting other types of cluster-errors. The excess
redundancy will be similar to the one of the constructions in
Section~\ref{sec:multi}.

\subsection{The coloring method}

The constructions with binary component codes use $D$ components
codes from which the first one is a burst-correcting code and the
other $D-1$ codes are burst-locator codes. Position $k$ in
component code $s$ is the binary summation of certain positions in
the array, which were defined with correspondence to some modulo
value related to $s$, $k$, and the $D$ indices which define the
position in the array. We generalize this idea to handle more
complicated cluster-errors to a method which will be called the
{\it coloring method}. A codeword is a $D$-dimensional array $\cA$
(not necessarily a $D$-dimensional box). We want to correct any
cluster-error with a given shape whose volume is $B$.

Again, we use $D$ binary component codes to correct the
cluster-error. The first code is a $(B+\delta_1)$-burst-correcting
code, $\delta_1 \geq 0$. The $s$-th component code, $2 \leq s \leq
D$, is a $(B+\delta_s)$-burst-locator code, $\delta_s \geq
\delta_1$. We further use $D$ different colorings of the
$D$-dimensional array. To each position of $\cA$ we assign a color
for each one of the $D$ colorings. Each coloring will be
associated with a different binary component code. For a given
coloring and the corresponding component code, position $k$ in the
component code is the binary sum of all bits which are colored
with color $k$. As we want to correct a cluster-error of a certain
shape in the array we want that the colorings will satisfy a few
properties:

\begin{itemize}
\item ({\bf p.1}) For the $s$-th coloring, for each $s$, $1 \leq s
\leq D$, the colors inside a burst of the given shape are distinct
integers and the difference between the largest integer and the
smallest one is at most $B +\delta_s -1$.

\item ({\bf p.2}) Given $D$ colorings, a color $\nu_s$ for the
$s$-th coloring, $1 \leq s \leq D$. There is at most one position
in the array which is colored with the colors $(\nu_1, \nu_2 ,
\ldots, \nu_D )$.

\item ({\bf p.3}) Any two positions which are colored with the
same color by the first coloring, have colors which differ by a
multiple of $B + \delta_s$ by the $s$-th coloring, for each $s$,
$2 \leq s \leq D$.
\end{itemize}

Finally, we have to choose redundancy bits in the array, in a
similar way to the method used in Section~\ref{sec:multi}.

\begin{theorem}
\label{thm:coloring} Assume that there exists a
$(B+\delta_1)$-burst-correcting code of length $n_1$ and for each
$s$, $2 \leq s \leq D$, there exists a
$(B+\delta_s)$-burst-locator code of length $n_s$. Assume further
that there exist $D$ colorings which satisfy properties (p.1),
(p.2), and (p.3), such that the $s$-th coloring assigns colors
between 1 to $n_s$ to the $D$-dimensional array. Then the coloring
method implies the existence of a
$B$-cluster-correcting code for a $D$-dimensional array $\cA$ and
a cluster with a given shape and volume $B$.
\end{theorem}
\begin{proof}
The proof is straightforward from the description. We just note,
that by property (p.2) there is no ambiguity in the erroneous
positions. By property (p.1), for each $s$, $1 \leq s \leq D$, the
erroneous positions affect at most $B+\delta_s$ consecutive
positions in the codeword of the $s$-th component code. Finally,
since two positions in the array which are assigned the same color
by the first coloring, have been assigned by the $s$-th coloring
colors which differ by a multiple of $B+\delta_s$ (see property
(p.3)), it follows that the possible bursts in the $s$-th
($B+\delta_s$)-burst-locator code are cyclic shifts of a burst
with length $B+\delta_s$.
\end{proof}

It will be more convenient if each coloring is a linear function
of the coordinate indices, i.e., given a position $(i_1,i_2,
\ldots , i_D )$ its color for the $s$-th coloring will be defined
by
\begin{align*}
\sum_{j=1}^D \alpha_j^s i_j
\end{align*}
where $\alpha_j^s$ is a constant integer which depends on the
coloring $s$ and the shape of the $D$-dimensional cluster. Such a
coloring will be called a {\it linear coloring}. With a linear
coloring we associate a {\it coloring matrix} $A_D$, where $( A_D
)_{s,j} = \alpha_j^s$. It is easy to verify that property (p.2),
is fulfilled for a linear coloring if and only if the coloring
matrix is an invertible matrix.

One can verify that the coloring method is a generalization of the
method described in Section~\ref{sec:multi}. To observe this we
should assign a color to position $(i_1,i_2, \ldots , i_D )$ by
the $s$-th coloring, $1 \leq s \leq D$, in a slightly different
way than the assignment in the next subsection. If position $k$ in
component code $s$ is the binary sum of a certain set $\cS$ of
positions in the array, then all positions of $\cS$ are colored
with color $k$ by the $s$-th coloring. We leave the exact
definitions of the colorings as an exercise for the reader.

\subsection{Multidimensional box-errors}

To demonstrate how the coloring method works we will first show
how it is used to correct multidimensional box-errors, where the
volume of the box error is an even integer.

Assume we want to construct an $n_1 \times n_2 \times \cdots
\times n_D$ $D$-dimensional ($b_1 \times b_2 \times \cdots \times
b_D$)-cluster-correcting code, where $B=\prod_{i=1}^D b_i$ is an
even integer. We will use $D$ binary component codes. One
component code will be able to correct a burst of length $B$ and
$D-1$ component codes will be able to locate the position of a
burst, whose length is $B+1$, given by a cyclic shift. Let

$$
B_j = b_j B_{j-1}, ~ \text{where}~ 1 \leq j \leq D,~ \text{and} ~
B_0 =1 ~.
$$

For each entry $(i_1 , i_2 , \ldots , i_D )$ in the array we
assign $D$ colors. The $s$-th color, $1\leq s\leq D$ is defined by
$$a_{i_1i_2\cdots i_D}^s
=\sum_{j=1}^{s-1}-B_{j-1}\frac{B_D}{B_{s-1}}i_j+
\sum_{j=s}^D\frac{B_{j-1}}{B_{s-1}}i_j.$$ Each coloring
corresponds to one component code. Codeword of component code $s$,
$\cC_s$, $1\le s\leq D$ is defined according to coloring $s$.
Position $k$ in the codeword is the sum modulo 2 of the values in
positions colored with color $k$ by the $s$-th coloring. We will
prove that these $D$ codes satisfy properties (p.1) through (p.3).

\begin{lemma}
\label{lem:p1_part1} If the $D$-dimensional code has a box error
of size $b_1 \times b_2 \times \cdots \times b_D$ then each one of
the $D$ component codewords has a burst whose length is at most
$B$.
\end{lemma}

\begin{proof}
A cluster occurred in the array is located inside a
multidimensional box of the form
$\{(i_1^*+i_1,i_2^*+i_2,\ldots,i_D^*+i_D) \ : \ 0\leq i_j\leq
b_j-1, 1\leq j\leq D\}$, for a fixed position $(
i_1^*,i_2^*,\ldots,i_D^* )$. The smallest color, $\ell_{\code_s}$,
of an erroneous position in $\cC_s$, $1\leq s\leq D$ is
$$\ell_{\code_s} =\sum_{j=1}^{s-1}-B_{j-1}\frac{B_D}{B_{s-1}}(i_j^*+b_j-1)+
\sum_{j=s}^D\frac{B_{j-1}}{B_{s-1}}i_j^*,$$ which is the color of
position $( i_1^* +b_1 -1,\ldots, i_{s-1}^* + b_{s-1} -1 , i_s^* ,
\ldots , i_D^* )$ and the largest color, $h_{\code_s}$, of an
erroneous position in $\cC_s$, $1\leq s\leq D$ is
$$h_{\code_s} =\sum_{j=1}^{s-1}-B_{j-1}\frac{B_D}{B_{s-1}}i_j^*+
\sum_{j=s}^D\frac{B_{j-1}}{B_{s-1}}(i_j^*+b_j-1),$$ which is the
color of position $( i_1^*,\ldots, i_{s-1}^* , i_s^* +b_s -1 ,
\ldots , i_D^* +b_D -1 )$. Now, we compute the difference
$h_{\cC_s} - \ell_{\cC_s}$.
\begin{align*}
h_{\code_s}-\ell_{\code_s}
=\sum_{j=1}^{s-1}B_{j-1}\frac{B_D}{B_{s-1}}(b_j-1)+
\sum_{j=s}^D\frac{B_{j-1}}{B_{s-1}}(b_j-1) \end{align*}
\begin{align*}
=\frac{B_D}{B_{s-1}}\Big(\sum_{j=1}^{s-1}B_{j-1}(b_j-1)\Big) +
\frac{1}{B_{s-1}}\Big(\sum_{j=s}^DB_{j-1}(b_j-1)\Big) \end{align*}
\begin{align*} = \frac{B_D}{B_{s-1}}\Big(B_{s-1}-B_0\Big) +
\frac{1}{B_{s-1}}\Big(B_D-B_{s-1}\Big) \end{align*} \begin{align*}
= B_D - \frac{B_D}{B_{s-1}}+\frac{B_D}{B_{s-1}}-1 = B-1
\end{align*}
Therefore, the length of a burst in each component code is at most
$B$.
\end{proof}

\begin{lemma}
\label{lem:p1_part2} For each one of the $D$ colorings, the $B$
colors in each $D$-dimensional box of size $b_1\times b_2\times
\cdots \times b_D$ in the array are all distinct.
\end{lemma}
\begin{proof}
Assume the contrary, that there exist two different positions
$(i_1,\ldots,i_D),(t_1,\ldots,t_D)$ located inside a box of size
$b_1\times b_2\times \cdots \times b_D$ in the array whose $s$-th
color is identical. Therefore, by definition,
\begin{align*}
\sum_{j=1}^{s-1}-B_{j-1}\frac{B_D}{B_{s-1}}i_j+
\sum_{j=s}^D\frac{B_{j-1}}{B_{s-1}}i_j
\end{align*}
\begin{align}
\label{eq:eq1} =\sum_{j=1}^{s-1}-B_{j-1}\frac{B_D}{B_{s-1}}t_j+
\sum_{j=s}^D\frac{B_{j-1}}{B_{s-1}}t_j
\end{align}
which implies
\begin{align*}
t_s-i_s =\sum_{j=1}^{s-1}-B_{j-1}\frac{B_D}{B_{s-1}}(i_j-t_j)+
\sum_{j=s+1}^D\frac{B_{j-1}}{B_{s-1}}(i_j-t_j)
\end{align*}
\begin{align}
\label{eq:eq2}
=b_s\Big(\sum_{j=1}^{s-1}-B_{j-1}\frac{B_D}{B_{s}}(i_j-t_j)+
\sum_{j=s+1}^D\frac{B_{j-1}}{B_{s}}(i_j-t_j)\Big)
\end{align}
Since these two positions are located inside a box of size $b_1
\times \cdots \times b_D$ it follows that $0\leq |t_s-i_s|\leq
b_s-1$ and hence $t_s-i_s=0$. We continue with (\ref{eq:eq2}) and
by induction we prove similarly that $t_k-i_k=0$ for each $k$,
$s+1 \leq k \leq D$. Therefore by (\ref{eq:eq1}) we have
\begin{align*}
\sum_{j=1}^{s-1}-B_{j-1}\frac{B_D}{B_{s-1}}(i_j-t_j) = 0
\end{align*}
which implies
\begin{align*}
\sum_{j=1}^{s-1}-B_{j-1}(i_j-t_j) = 0~.
\end{align*}
Now, we will show by induction that for each $j$, $1\leq j \leq
s-1$, $t_j-i_j=0$. For $j=1$, we have
\begin{align*}
t_1-i_1=\sum_{j=2}^{s-1}B_{j-1}(i_j-t_j)=b_1\big(\sum_{j=2}^{s-1}\frac{B_{j-1}}{b_1}(i_j-t_j)
\big)
\end{align*}
and since $0\leq |t_1-i_1|\leq b_1-1$ we have $t_1-i_1=0$. We
continue in similar way and obtain for each $j$, $1\leq j \leq
s-1$, $t_j-i_j=0$. Therefore, for each $1\leq s \leq D$, we have
$t_s=i_s$. Thus, for each one of the $D$ colorings, the $B$ colors
in each $D$-dimensional box of size $b_1\times b_2\times \cdots
\times b_D$ in the array are all distinct.
\end{proof}

\begin{lemma}
\label{lem:multipleB+1} Any two positions which are colored with
the same color by the first coloring, have colors which differ by
a multiple of $B+1$ by the the $s$-th coloring, for each $s$,
$2\leq s\leq D$.
\end{lemma}
\begin{proof}
Assuming the $k$-th position in the codeword of $\code_1$ is
erroneous. This error results from an array error in position
$(i_1,i_2,\ldots,i_D)$ such that
$a_{i_1,i_2,\ldots,i_D}^1=\sum_{j=1}^DB_{j-1}i_j=k$, and hence
$i_1=k-\sum_{j=2}^DB_{j-1}i_j$. The possible error locations in
$\code_s$, $2\leq s \leq D$ are of the form
\begin{eqnarray}
& & a_{i_1i_2\cdots i_D}^s
=\sum_{j=1}^{s-1}-B_{j-1}\frac{B_D}{B_{s-1}}i_j+
\sum_{j=s}^D\frac{B_{j-1}}{B_{s-1}}i_j= \nonumber \\
& &
-\frac{B_D}{B_{s-1}}(k-\sum_{j=2}^DB_{j-1}i_j)+\sum_{j=2}^{s-1}-B_{j-1}\frac{B_D}{B_{s-1}}i_j+
\nonumber \\
& & \sum_{j=s}^D\frac{B_{j-1}}{B_{s-1}}i_j = \nonumber \\
& & -\frac{B_D}{B_{s-1}}k
+\sum_{j=2}^DB_{j-1}\frac{B_D}{B_{s-1}}i_j-\sum_{j=2}^{s-1}B_{j-1}\frac{B_D}{B_{s-1}}i_j+\nonumber
\\
& & \sum_{j=s}^D\frac{B_{j-1}}{B_{s-1}}i_j = \nonumber \\
& & -\frac{B_D}{B_{s-1}}k +
\sum_{j=s}^D\big(B_{j-1}\frac{B_D}{B_{s-1}}+\frac{B_{j-1}}{B_{s-1}}\big)i_j
= \nonumber \\
& & -\frac{B_D}{B_{s-1}}k +
(B+1)\sum_{j=s}^D\frac{B_{j-1}}{B_{s-1}}i_j. \nonumber
\end{eqnarray}
$\frac{B_D}{B_{s-1}}k$ is a constant and therefore, two positions
which have the same color by the first coloring, have colors which
differ in a multiple of $B+1$ by the $s$-th coloring, $2 \leq s
\leq D$.
\end{proof}

\begin{lemma}
\label{lem:exact_position} Given a position $c_s$ in the $s$-th
component code, $1 \leq s \leq D$, the set of equations
\begin{equation}
\label{equation set} c_s=a_{i_1i_2\cdots i_D}^s, 1\leq s \leq D
\end{equation} has exactly one
solution for $(i_1 , \ldots , i_D )$.
\end{lemma}

\begin{proof}
We will prove that the coloring matrix is invertible by proving
that its determinant is not equal to {\it zero}. The $(s,j)$ entry
of the coloring matrix $A_D$ is given by
\begin{displaymath}
(A_D)_{s,j} = \left\{ \begin{array}{ll}
-B_{j-1}\frac{B_D}{B_{s-1}} & \textrm{for $j<s$}\\
\frac{B_{j-1}}{B_{s-1}} & \textrm{for $j\geq s$}
\end{array} \right. .
\end{displaymath}
We will prove by induction on $D$, $D \geq 2$, that $|A_D|=(1+B)^{D-1}$,
where $B$ is the volume of the box-error.\\
For the basis of the induction, $A_2=\left(%
\begin{array}{cc}
  1 & b_1 \\
  -b_2 & 1 \\
\end{array}%
\right)$ and hence $|A_2|=1+b_1b_2 =1+b_1 b_2$, where the
two-dimensional
cluster has size $b_1 \times b_2$.\\
For the induction hypothesis we assume that the determinant of
every coloring matrix
of size $(D-1)\times (D-1)$ is given by $|A_{D-1}|=(1+B' )^{D-2}$, where $B'$
is the volume of the corresponding $(D-1)$-dimensional box-error.\\
For the induction step let $A_D$ be a coloring matrix of size
$D\times D$. The determinant of $A_D$ is given by
$$|A_D|=\sum_{j=1}^D(-1)^{j+1}B_{j-1}|A_D
[{1,j}]|,$$ where $A_D [{1,j}]$ is the matrix obtained from $A_D$
by deleting row 1 and column $j$. For $2\leq s\leq D-1$ the $s$-th
row of $A_D [ {1,s}]$ is given by:
$$\frac{1}{B_{s-1}}(-B_DB_{0},-B_DB_{1},\ldots,-B_DB_{s-2},B_{s},\ldots,B_{D-1}),$$
and the $(s+1)$-th row is given by:
$$\frac{1}{B_{s}}(-B_DB_{0},-B_DB_{1},\ldots,-B_DB_{s-2},B_{s},\ldots,B_{D-1}).$$
These two rows are linearly dependent and therefore,
$|(A_D)_{1s}|=0$ for $2 \leq s \leq D-1$ and we have
\begin{align}
\label{eq:detA_D}
|A_D|=B_{0}|(A_D)_{11}|+(-1)^{D+1}B_{D-1}|(A_D)_{1D}|.
\end{align}
The matrix $(A_D)_{11}$ is also a coloring matrix with respect to
coloring related to the box-error $b_2 \times b_3 \times \cdots
\times b_{D-1}\times (b_{D}b_1)$ and according to the induction
assumption its determinant is given by $(A_D)_{11}=(1+B)^{D-2}$.
Let $(A_D')_{1D}$ be the matrix constructed from $(A_D)_{1D}$ by
dividing each element of the last row of $(A_D)_{1D}$ by $-b_D$
and shifting this row cyclically to be the first row. Clearly,
$|(A_D)_{1D}|=-b_D (-1)^{D-2} | (A_D')_{1D}|$. The matrix
$(A_D')_{1D}$ is also a coloring matrix with respect to coloring
related to the error box $b_1 \times b_2 \times \cdots \times
b_{D-2} \times (b_{D-1}b_D)$. Hence, $|(A_D)_{1D}|=-b_D
(-1)^{D-2}(1+B)^{D-2}$ and from (\ref{eq:detA_D}) we have
$|A_D|=(1+B)^{D-2}+(-1)^{D+1}B_{D-1}(-b_D)(-1)^{D-2}(1+B)^{D-2} =
(1+B)^{D-1}$. Thus, $A_D$ is invertible and (\ref{equation set})
has a unique solution.
\end{proof}

The encoding procedure is quite straight forward. First we have to
choose three sets [or $D+1$ sets] of positions for the redundancy
bits as in the previous constructions. Position $k$ in the
codeword of the $s$-th component code is the binary sum of all
positions in the array colored with $k$ by the $s$-th coloring. In
the decoding procedure, the first component code provides a list
of erroneous positions. If $k$ is the erroneous location of the
first code, the difference between equivalent errors in $\cC_s$,
$2\leq s\leq D$ is a multiple of $B+1$ and by
Lemma~\ref{lem:multipleB+1} the value $-\frac{B_D}{B_{s-1}}k$ is
the residue modulo $B+1$ of the error location in the codeword of
$\cC_s$. Therefore, for the codeword of $\cC_s$ the burst is known
up to a cyclic permutation of length $B+1$. $\cC_s$ is a
burst-locator code which can locate where such burst started.
Hence, each burst-locator code locates the position in which the
burst has started with respect to the corresponding dimension.
Therefore, the location of the erroneous position $k$ in the first
code is known in each of the $D-1$ burst-locator codes. Hence, we
can partition the positions of the errors in the $D$ component
codewords into $D$-tuples. Each such tuple is of the form
$(c_1,c_2,\ldots,c_D)$, where $c_s$ is the error location in the
codeword of $\cC_s$, $1\leq s \leq D$. By
Lemma~\ref{lem:exact_position} such a $D$-tuple corresponds to
exactly one position in the $D$-dimensional array, which is an
erroneous position.

To summarize, by Lemmas~\ref{lem:p1_part1} and~\ref{lem:p1_part2}
the $D$ colorings satisfy (p.1), by Lemma~\ref{lem:multipleB+1}
they satisfy (p.2), and by Lemma~\ref{lem:exact_position} they
satisfy (p.3). Hence, by Theorem~\ref{thm:coloring} the code is a
$(b_1 \times b_2 \times \cdots \times b_D)$-cluster-correcting
code. The redundancy of the code is slightly larger than the one
for odd $B$ (see Theorem~\ref{thm:redundancyD}). We omit the
tedious proof.


\section{Lee Spheres Cluster Errors}
\label{sec:Lee}

An error event at a position $(i_1 , i_2 , \ldots , i_D)$ can
affect other positions around it. If we assume that the error
event can be spread up to radius $R$, then any position $(t_1 ,
t_2 , \ldots , t_D)$ such that $\sum_{\ell=1}^D \abs{t_\ell
-i_\ell} \leq R$ might be erroneous. The set of positions $\{ (t_1
, t_2 , \ldots , t_D) ~:~ \sum_{\ell=1}^D \abs{t_\ell -i_\ell}
\leq R \}$ forms a $D$-dimensional Lee sphere with radius
$R$~\cite{BBV,GoWe}. Another important observation is that any
arbitrary cluster-error of size $b$ is located inside a Lee sphere
with radius $\lfloor \frac{b}{2} \rfloor$. We wish to supply a
method to correct all errors which might occur in such a given Lee
sphere. The size of such sphere is $\sum_{j=0}^{\text{min}\{D,R\}}
2^j {D \choose j}{R \choose j}$~\cite{GoWe} and each position can
be erroneous. We can think about three different methods to
perform the task.

\begin{itemize}
\item We can use any method which is able to correct all errors
occurred in a $\underbrace{(2R+1)\times \cdots \times
(2R+1)}_{D\textrm{ times}}$ $D$-dimensional cube since a Lee
sphere with radius $R$ is located inside such a cube. The main
disadvantage of this method is that the code can correct much more
errors than the ones needed in its goal. For example, the size of
the $D$-dimensional cube is $(2R+1)^D$, while the size of the
$D$-dimensional Lee sphere with radius $R$ is $\frac{(2R)^D}{D!} +
O(R^{D-1})$ when $D$ is fixed and $R \longrightarrow \infty$.

\item We can transform the space into another space in a way that
a Lee sphere will be transformed to a shape located inside another
error-shape with a small volume which we know how to correct
efficiently.

\item We can find a coloring and perform correction of errors as
explained in Section~\ref{sec:coloring}.

\end{itemize}
In this section we will explain how we apply the last two methods.

\subsection{Space transformation}

The idea we are going to use is to transform the space into
another space in such a way that a Lee sphere in one space will be
transformed into a shape located inside a box in the second space.
After that we will be able to use the encoding and decoding
introduced in Section~\ref{sec:multi}. We start with the
two-dimensional transformation.

\begin{lemma}
\label{Lee_Loc_Rec} Let $M,M^*$ be infinite two-dimensional arrays
and let $T$ be the transformation from $M$ into $M^*$ defined by
$T\roundb{i_1,i_2} = \roundb{\lceil\frac{i_1+i_2}{2}\rceil,
i_2-i_1}$. Then a Lee sphere with radius $R$ in the array $M$ is
located after the transformation $T$ inside a rectangle of size
$(R+1)\times (2R+1)$ in $M^*$.
\end{lemma}
\begin{proof}
A Lee sphere of radius $R$ with center at $\roundb{i_i,i_2}$ in
the array $M$ includes the set of positions $B_R\roundb{i_1,i_2} =
\{ \roundb{t_1,t_2} \ ~:~ \ |t_1-i_1| + |t_2-i_2| \leq R \} = \{
\roundb{i_1 + R_\ell ,i_2 + R_j} \ ~:~ \ |R_\ell| + |R_j| \leq R
\}. $ $B_R\roundb{i_1,i_2}$ is transformed by $T$ into the set of
positions $B_R^*\roundb{i_1,i_2} = \left\{
\roundb{\left\lceil\frac{i_1+R_\ell+i_2+R_j}{2}\right\rceil,i_2+R_j-i_1-R_\ell}
\  : \ |R_\ell| + |R_j| \leq R \right\}$ of $M^*$. Denote,
$i_1^*=\left\lceil\frac{i_1+i_2-R}{2}\right\rceil, \ i_2^* =
i_2-i_1-R$, and we have that $B_R^*\roundb{i_1,i_2}$ is located
inside the rectangle $\left\{ \roundb{i_1^*+ t_1,i_2^*+t_2} \ ~:~
\ 0\leq \ t_1 \leq R, 0\leq \ t_2 \leq 2R \right\}$.
\end{proof}

The transformation $T$ transforms a parallelogram into a rectangle
(see Fig.~\ref{parallelogram} and~\ref{rectangle}) and hence we
will need some adjustment in our encoding and decoding procedures
if we want to correct Lee sphere clusters in a rectangular array
rather than a parallelogram.

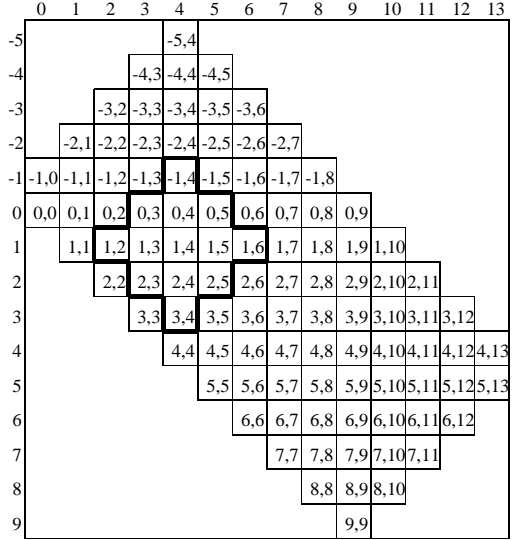
\begin{figure}[htbp]
\begin{center}
\setlength{\unitlength}{0.46cm}
\begin{picture}(-2,15)(8,0)

 \linethickness{0.075mm} \scriptsize{
\put(0,0){\line(0,1){15}} \put(1,12){\line(0,-1){4}}
\put(2,13){\line(0,-1){6}} \put(3,14){\line(0,-1){8}}
\put(4,15){\line(0,-1){10}} \put(5,15){\line(0,-1){11}}
\put(6,14){\line(0,-1){11}} \put(7,13){\line(0,-1){11}}
\put(8,12){\line(0,-1){11}} \put(9,11){\line(0,-1){11}}
\put(10,10){\line(0,-1){10}} \put(11,9){\line(0,-1){8}}
\put(12,8){\line(0,-1){6}} \put(13,7){\line(0,-1){4}}
\put(14,0){\line(0,1){15}}

\put(0,15){\line(1,0){14}}      \put(3,14){\line(1,0){3}}
\put(2,13){\line(1,0){5}}      \put(1,12){\line(1,0){7}}
\put(0,11){\line(1,0){9}}     \put(0,10){\line(1,0){10}}
\put(0,9){\line(1,0){11}}     \put(1,8){\line(1,0){11}}
\put(2,7){\line(1,0){11}}     \put(3,6){\line(1,0){11}}
\put(4,5){\line(1,0){10}}    \put(5,4){\line(1,0){9}}
\put(6,3){\line(1,0){7}}      \put(7,2){\line(1,0){5}}
\put(8,1){\line(1,0){3}}    \put(0,0){\line(1,0){14}}

\put(0.25,9.25){0,0} \put(0.115,10.25){-1,0}
\put(1.115,10.25){-1,1} \put(1.115,11.25){-2,1}
\put(2.115,11.25){-2,2} \put(2.115,12.25){-3,2}
\put(3.115,12.25){-3,3} \put(3.115,13.25){-4,3}
\put(4.115,13.25){-4,4} \put(4.115,14.25){-5,4}

\put(1.25,8.25){1,1} \put(1.25,9.25){0,1} \put(2.25,9.25){0,2}
\put(2.115,10.25){-1,2} \put(3.115,10.25){-1,3}
\put(3.115,11.25){-2,3} \put(4.115,11.25){-2,4}
\put(4.115,12.25){-3,4} \put(5.115,12.25){-3,5}
\put(5.115,13.25){-4,5}

\put(2.25,7.25){2,2} \put(2.25,8.25){1,2} \put(3.25,8.25){1,3}
\put(3.25,9.25){0,3} \put(4.25,9.25){0,4} \put(4.115,10.25){-1,4}
\put(5.115,10.25){-1,5} \put(5.115,11.25){-2,5}
\put(6.115,11.25){-2,6} \put(6.115,12.25){-3,6}

\put(3.25,6.25){3,3} \put(3.25,7.25){2,3} \put(4.25,7.25){2,4}
\put(4.25,8.25){1,4} \put(5.25,8.25){1,5} \put(5.25,9.25){0,5}
\put(6.25,9.25){0,6} \put(6.115,10.25){-1,6}
\put(7.115,10.25){-1,7} \put(7.115,11.25){-2,7}

\put(4.25,5.25){4,4} \put(4.25,6.25){3,4} \put(5.25,6.25){3,5}
\put(5.25,7.25){2,5} \put(6.25,7.25){2,6} \put(6.25,8.25){1,6}
\put(7.25,8.25){1,7} \put(7.25,9.25){0,7} \put(8.25,9.25){0,8}
\put(8.115,10.25){-1,8}

\put(5.25,4.25){5,5} \put(5.25,5.25){4,5} \put(6.25,5.25){4,6}
\put(6.25,6.25){3,6} \put(7.25,6.25){3,7} \put(7.25,7.25){2,7}
\put(8.25,7.25){2,8} \put(8.25,8.25){1,8} \put(9.25,8.25){1,9}
\put(9.25,9.25){0,9}

\put(6.25,3.25){6,6} \put(6.25,4.25){5,6} \put(7.25,4.25){5,7}
\put(7.25,5.25){4,7} \put(8.25,5.25){4,8} \put(8.25,6.25){3,8}
\put(9.25,6.25){3,9} \put(9.25,7.25){2,9} \put(10.06,7.25){2,10}
\put(10.06,8.25){1,10}

\put(7.25,2.25){7,7} \put(7.25,3.25){6,7} \put(8.25,3.25){6,8}
\put(8.25,4.25){5,8} \put(9.25,4.25){5,9} \put(9.25,5.25){4,9}
\put(10.06,5.25){4,10} \put(10.06,6.25){3,10}
\put(11.06,6.25){3,11} \put(11.06,7.25){2,11}

\put(8.25,1.25){8,8} \put(8.25,2.25){7,8} \put(9.25,2.25){7,9}
\put(9.25,3.25){6,9} \put(10.06,3.25){6,10} \put(10.06,4.25){5,10}
\put(11.06,4.25){5,11} \put(11.06,5.25){4,11}
\put(12.06,5.25){4,12} \put(12.06,6.25){3,12}

\put(9.25,0.25){9,9} \put(9.25,1.25){8,9} \put(10.06,1.25){8,10}
\put(10.06,2.25){7,10} \put(11.06,2.25){7,11}
\put(11.06,3.25){6,11} \put(12.06,3.25){6,12}
\put(12.06,4.25){5,12} \put(13.06,4.25){5,13}
\put(13.06,5.25){4,13}

\put(-0.35,0.25){9} \put(-0.35,1.25){8} \put(-0.35,2.25){7}
\put(-0.35,3.25){6} \put(-0.35,4.25){5} \put(-0.35,5.25){4}
\put(-0.35,6.25){3} \put(-0.35,7.25){2} \put(-0.35,8.25){1}
\put(-0.35,9.25){0} \put(-0.47,10.25){-1} \put(-0.47,11.25){-2}
\put(-0.47,12.25){-3} \put(-0.47,13.25){-4} \put(-0.47,14.25){-5}

\put(0.35,15.15){0} \put(1.35,15.15){1} \put(2.35,15.15){2}
\put(3.35,15.15){3} \put(4.35,15.15){4} \put(5.35,15.15){5}
\put(6.35,15.15){6} \put(7.35,15.15){7} \put(8.35,15.15){8}
\put(9.35,15.15){9} \put(10.35,15.15){10} \put(11.35,15.15){11}
\put(12.35,15.15){12} \put(13.35,15.15){13} }
\linethickness{0.5mm} \put(2,8){\line(0,1){1}}
\put(2,8){\line(1,0){1}} \put(2,9){\line(1,0){1}}
\put(3,9){\line(0,1){1}} \put(3,10){\line(1,0){1}}
\put(4,10){\line(0,1){1}} \put(4,11){\line(1,0){1}}
\put(5,10){\line(0,1){1}} \put(5,10){\line(1,0){1}}
\put(6,9){\line(0,1){1}} \put(6,9){\line(1,0){1}}
\put(7,8){\line(0,1){1}} \put(6,8){\line(1,0){1}}
\put(6,7){\line(0,1){1}} \put(5,7){\line(1,0){1}}
\put(5,6){\line(0,1){1}} \put(4,6){\line(1,0){1}}
\put(4,6){\line(0,1){1}} \put(3,7){\line(1,0){1}}
\put(3,7){\line(0,1){1}}

\end{picture}

\caption{Positions in the parallelogram}%
\label{parallelogram}

\end{center}
\end{figure}

\begin{figure}[htbp]
\begin{center}
\setlength{\unitlength}{0.46cm}
\begin{picture}(-5.5,9)(8,0)

\linethickness{0.075mm}
\multiput(0,0)(1,0){11}%
{\line(0,1){10}}
\multiput(0,0)(0,1){11}%
{\line(1,0){10}} \scriptsize{ \put(0.25,9.25){0,0}
\put(1.115,9.25){-1,0} \put(2.115,9.25){-1,1}
\put(3.115,9.25){-2,1} \put(4.115,9.25){-2,2}
\put(5.115,9.25){-3,2} \put(6.115,9.25){-3,3}
\put(7.115,9.25){-4,3} \put(8.115,9.25){-4,4}
\put(9.115,9.25){-5,4}

\put(0.25,8.25){1,1} \put(1.25,8.25){0,1} \put(2.25,8.25){0,2}
\put(3.115,8.25){-1,2} \put(4.115,8.25){-1,3}
\put(5.115,8.25){-2,3} \put(6.115,8.25){-2,4}
\put(7.115,8.25){-3,4} \put(8.115,8.25){-3,5}
\put(9.115,8.25){-4,5}

\put(0.25,7.25){2,2} \put(1.25,7.25){1,2} \put(2.25,7.25){1,3}
\put(3.25,7.25){0,3} \put(4.25,7.25){0,4} \put(5.115,7.25){-1,4}
\put(6.115,7.25){-1,5} \put(7.115,7.25){-2,5}
\put(8.115,7.25){-2,6} \put(9.115,7.25){-3,6}

\put(0.25,6.25){3,3} \put(1.25,6.25){2,3} \put(2.25,6.25){2,4}
\put(3.25,6.25){1,4} \put(4.25,6.25){1,5} \put(5.25,6.25){0,5}
\put(6.25,6.25){0,6} \put(7.115,6.25){-1,6} \put(8.115,6.25){-1,7}
\put(9.115,6.25){-2,7}

\put(0.25,5.25){4,4} \put(1.25,5.25){3,4} \put(2.25,5.25){3,5}
\put(3.25,5.25){2,5} \put(4.25,5.25){2,6} \put(5.25,5.25){1,6}
\put(6.25,5.25){1,7} \put(7.25,5.25){0,7} \put(8.25,5.25){0,8}
\put(9.115,5.25){-1,8}

\put(0.25,4.25){5,5} \put(1.25,4.25){4,5} \put(2.25,4.25){4,6}
\put(3.25,4.25){3,6} \put(4.25,4.25){3,7} \put(5.25,4.25){2,7}
\put(6.25,4.25){2,8} \put(7.25,4.25){1,8} \put(8.25,4.25){1,9}
\put(9.25,4.25){0,9}

\put(0.25,3.25){6,6} \put(1.25,3.25){5,6} \put(2.25,3.25){5,7}
\put(3.25,3.25){4,7} \put(4.25,3.25){4,8} \put(5.25,3.25){3,8}
\put(6.25,3.25){3,9} \put(7.25,3.25){2,9} \put(8.06,3.25){2,10}
\put(9.06,3.25){1,10}

\put(0.25,2.25){7,7} \put(1.25,2.25){6,7} \put(2.25,2.25){6,8}
\put(3.25,2.25){5,8} \put(4.25,2.25){5,9} \put(5.25,2.25){4,9}
\put(6.06,2.25){4,10} \put(7.06,2.25){3,10} \put(8.06,2.25){3,11}
\put(9.06,2.25){2,11}

\put(0.25,1.25){8,8} \put(1.25,1.25){7,8} \put(2.25,1.25){7,9}
\put(3.25,1.25){6,9} \put(4.06,1.25){6,10} \put(5.06,1.25){5,10}
\put(6.06,1.25){5,11} \put(7.06,1.25){4,11} \put(8.06,1.25){4,12}
\put(9.06,1.25){3,12}

\put(0.25,0.25){9,9} \put(1.25,0.25){8,9} \put(2.06,0.25){8,10}
\put(3.06,0.25){7,10} \put(4.06,0.25){7,11} \put(5.06,0.25){6,11}
\put(6.06,0.25){6,12} \put(7.06,0.25){5,12} \put(8.06,0.25){5,13}
\put(9.06,0.25){4,13} } \linethickness{0.5mm}
\put(1,8){\line(1,0){5}} \put(1,5){\line(0,1){3}}
\put(6,5){\line(0,1){3}} \put(1,5){\line(1,0){1}}
\put(2,6){\line(1,0){1}} \put(3,5){\line(1,0){1}}
\put(4,6){\line(1,0){1}} \put(5,5){\line(1,0){1}}

\put(2,5){\line(0,1){1}} \put(3,5){\line(0,1){1}}
\put(4,5){\line(0,1){1}} \put(5,5){\line(0,1){1}}

\end{picture}

\caption{Transformation to rectangle}%
\label{rectangle}

\end{center}
\end{figure}


The construction is generalized to $D$ dimensions. The
transformation $T$ will work between two $D$-dimensional arrays.
For each entry $\roundb{i_1,i_2,\ldots,i_D}$ in the
$D$-dimensional array $M$, $T\roundb{i_1,i_2,\ldots,i_D} =
\roundb{i_1^T,i_2^T,\ldots,i_D^T}$, where $ i_1^T =
\left\lceil\frac{i_1+i_2}{2}\right\rceil, \ i_2^T =
\left\lceil\frac{-i_1+i_2+i_3}{2}\right\rceil,\ldots,
i_j^T=\left\lceil\frac{(-1)^{j+1}i_1+(-1)^ji_2+\cdots-i_{j-1}+i_j+i_{j+1}}{2}\right\rceil,$
for $1\leq j \leq D-1$ and $i_D^T=\Sigma_{j=1}^D (-1)^{D-j} i_j$.
We invoke first the two-dimensional transformation $T$ on the
first two coordinates, then on the second and the third
coordinates and so on. Similarly to the proof of
Lemma~\ref{Lee_Loc_Rec} we can prove the following lemma.

\begin{lemma}
\label{Lee_Loc_Rec_Mult_Dim}  Let $M$ be an infinite
$D$-dimensional array. Then a Lee sphere with radius $R$ in the
array $M$ is located after the transformation $T$ inside a
$D$-dimensional box of size $\underbrace{(R+1)\times (R+1)\times
\cdots \times (R+1)}_{D-1\textrm{ times}} \times(2R+1)$ in $M^*$.
\end{lemma}

As a consequence of this transformation we can use the
constructions of Sections~\ref{sec:multi} and~\ref{sec:coloring}
for correction of $D$-dimensional box-error. If we assume that our
codewords are $D$-dimensional arrays rather than $D$-dimensional
parallelograms we can use an array located inside the
parallelogram. The redundancy in this case will be slightly
increased.

\subsection{Tiling and coloring}
\label{subsec:tiling}

When the error shape is a two-dimensional Lee sphere with radius
$R$ we can provide a code with a better redundancy than the code
constructed by using the two-dimensional space transformation. We
will use the coloring method of Section~\ref{sec:coloring}.
Indeed, this construction is a good example for efficient uses of
the coloring method. We choose two appropriate colorings $\Psi_1$
and $\Psi_2$. For a given position $(i_1,i_2)$ in the $n_1 \times
n_2$ two-dimensional array let $\Psi_1 (i_1,i_2) = (R+1)i_1+Ri_2$
and $\Psi_2 (i_1,i_2) = -Ri_1+(R+1)i_2$ (see Fig.~\ref{tiling}).

\begin{lemma}
\label{lem:multiple} If $\Psi_1 (i_1,i_2)=\Psi_1 (t_1,t_2)$ then
$\Psi_2 (i_1,i_2)-\Psi_2 (t_1,t_2)$ is a multiple of $2R^2+2R+1$.
\end{lemma}
\begin{proof}
$\Psi_1 (i_1,i_2)=\Psi_1 (t_1,t_2)$ implies that $(R+1)i_1+Ri_2 =
(R+1)t_1 +Rt_2$. Hence, $R(i_2-t_2)=(R+1)(t_1-i_1)$, i.e.,
$(R+1)(i_2-t_2)=\frac{(R+1)^2}{R}(t_1-i_1)$.

Now, $\Psi_2 (i_1,i_2)-\Psi_2 (t_1,t_2) = -Ri_1 + (R+1)i_2
-(-Rt_1+(R+1)t_2)=R(t_1-i_1)+(R+1)(i_2-t_2)=R(t_1-i_1)+
\frac{(R+1)^2}{R}(t_1-i_1)=\frac{2R^2+2R+1}{R}(t_1-i_1)$.

Since $R$ and $2R^2+2R+1$ are relatively primes it follows that
$R$ divides $t_1-i_1$ and hence $\Psi_2 (i_1,i_2)-\Psi_2
(t_1,t_2)$ is a multiple of $2R^2+2R+1$.
\end{proof}

The effect of these colorings of the two-dimensional arrays is
best seen if we consider a tiling of the two-dimensional space
with Lee spheres with radius $R$. Such tiling is well known and
given in~\cite{BBV,GoWe}. In this tiling each Lee sphere belongs
to two diagonal strips. For the first coloring all relative
positions of the Lee spheres in the same diagonal strip have the
same number; and in the other direction they are congruent modulo
$b^*=2R^2+2R+1$ which is the size of a sphere (see
Fig.~\ref{tiling}).

\begin{lemma}
\label{lem:consecutive} For any given Lee sphere with radius $R$
in a two-dimensional array, the color numbers assigned by
$\Psi_i$, $i=1,2$, to the positions of the Lee sphere are
$2R^2+2R+1$ consecutive integers.
\end{lemma}
\begin{proof}
Given a Lee sphere centered at the point $(i_1,i_2)$ it is readily
verified that smallest color number assigned by $\Psi_1$ to a
point in the Lee sphere is $\Psi_1(i_1-R,i_2)$ and the largest
color number is $\Psi_1(i_1+R,i_2)$. Now,
$\Psi_1(i_1+R,i_2)-\Psi_1(i_1-R,i_2)=(R+1)(i_1+R)+Ri_2 -
((R+1)(i_1-R)+Ri_2)=2R^2+2R$. Hence, to complete the proof we have
to show that all the color numbers assigned by $\Psi_1$ to the
positions inside a Lee sphere are distinct.

Let $(i_1,i_2)$ and $(t_1,t_2)$ be a pair of points for which
$\Psi_1 (i_1,i_2)=\Psi_1 (t_1,t_2)$. By the proof of
Lemma~\ref{lem:multiple} we have that $R(i_2-t_2)=(R+1)(t_1-i_1)$.
Hence, $i_2-t_2=\frac{R+1}{R}(t_1-i_1)$ and
$t_1-i_1=\frac{R}{R+1}(i_2-t_2)$. Since $R$ and $R+1$ are
relatively primes it follows that $R$ divides $t_1-i_1$ and $R+1$
divides $i_2-t_2$. This implies that $R \leq | t_1-i_1 |$ and $R+1
\leq | i_2 - t_2 |$. Therefore, $2R+1 \leq | t_1-i_1 | + | i_2 -
t_2 |$ and the two points $(i_1,i_2)$ and $(t_1,t_2)$ cannot be
contained in the same Lee sphere.

Thus, all color numbers assigned by $\Psi_1$, to the positions
contained in a Lee sphere are $2R^2+2R+1$ consecutive integers.
The same proof holds for $\Psi_2$.
\end{proof}

\begin{lemma}
\label{lem:invert} The coloring matrix defined by $\Psi_1$ and
$\Psi_2$ is an invertible matrix.
\end{lemma}
\begin{proof}
The coloring matrix defined by $\Psi_1$ and $\Psi_2$ is
\begin{displaymath}
\left(%
\begin{array}{cc}
  R+1 & R \\
  -R & R+1 \\
\end{array}%
\right)
\end{displaymath}
which is clearly an invertible matrix.
\end{proof}

By Lemmas~\ref{lem:multiple}, \ref{lem:consecutive}
,and~\ref{lem:invert}, we have that (p.1), (p.2), and (p.3), are
satisfied respectively. Hence, by Theorem~\ref{thm:coloring} the
code constructed is capable to correct a Lee-sphere error with
radius $R$ and size $b^*$.

We choose $r+m+1$ appropriate redundancy bits similarly to the
constructions in Section~\ref{sec:multi}. We use two component
codes. One code is a $b^*$-burst-correcting code of length
$2^{r-b^*+1}-1$, where $2^{r-b^*+1}$ is the least power of 2
greater than $(R+1)(n_1 -1)+R(n_2-1)+1$ which is the number of
colors needed to color the array with the first coloring. The
second component code is a $b^*$-burst-locator code of length
$2^m-1$ (such a code exists since $b^*=2R^2+2R+1$ is an odd
integer). We take $2^m$ to be the least power of 2 greater than
$(R+1)(n_2 -1)+R(n_1-1)+1$ which is the number of colors needed to
color the array with the second coloring (for the simplicity of
the computations we will take $2^{r-b^*+1} > (R+1)n_1 + Rn_2 \geq
2^{r-b^*}$ and $2^m > (R+1)n_2 +Rn_1 \geq 2^{m-1}$). The
redundancy of the $n_1 \times n_2$ two-dimensional code is at most
$\lceil \text{log}_2 (n_1n_2) \rceil + b^* + \lceil 2 \text{log}_2
(2R+1) \rceil +2$. If $n_1=n_2=n$ then the redundancy of the
two-dimensional code is at most $\lceil \text{log}_2 n^2 \rceil +
b^* +\lceil 2\text{log}_2(2R+1) \rceil$ and if each codeword is a
rhombus of size $n^2$ then the redundancy is at most $\lceil
\text{log}_2 n^2 \rceil + b^* +\lceil \text{log}_2 b^* \rceil$.

\begin{figure*}
\begin{center}
\mbox{
    \subfigure[]{ \setlength{\unitlength}{0.32cm}
\begin{picture}(25.5,24)(0,0)
\linethickness{0.01mm}
\multiput(4,4)(1,0){17}%
{\line(0,1){16}}
\multiput(4,4)(0,1){17}%
{\line(1,0){16}}

\linethickness{0.6mm}

\put(0,16){\line(0,1){1}}

\put(1,15){\line(0,1){1}} \put(1,8){\line(0,1){1}}
\put(1,17){\line(0,1){1}}

\put(2,7){\line(0,1){1}} \put(2,9){\line(0,1){1}}
\put(2,13){\line(0,1){2}} \put(2,18){\line(0,1){1}}

\put(3,5){\line(0,1){2}} \put(3,10){\line(0,1){1}}
\put(3,12){\line(0,1){1}} \put(3,14){\line(0,1){1}}
\put(3,18){\line(0,1){1}}

\put(4,4){\line(0,1){1}} \put(4,6){\line(0,1){1}}
\put(4,10){\line(0,1){2}} \put(4,15){\line(0,1){1}}
\put(4,17){\line(0,1){1}} \put(4,19){\line(0,1){1}}

\put(5,2){\line(0,1){2}} \put(5,7){\line(0,1){1}}
\put(5,9){\line(0,1){1}} \put(5,11){\line(0,1){1}}
\put(5,15){\line(0,1){2}} \put(5,20){\line(0,1){1}}

\put(6,1){\line(0,1){1}} \put(6,3){\line(0,1){1}}
\put(6,7){\line(0,1){2}} \put(6,12){\line(0,1){1}}
\put(6,14){\line(0,1){1}} \put(6,16){\line(0,1){1}}
\put(6,20){\line(0,1){1}}

\put(7,0){\line(0,1){1}} \put(7,4){\line(0,1){1}}
\put(7,6){\line(0,1){1}} \put(7,8){\line(0,1){1}}
\put(7,12){\line(0,1){2}} \put(7,17){\line(0,1){1}}
\put(7,19){\line(0,1){1}} \put(7,21){\line(0,1){1}}

\put(8,0){\line(0,1){1}} \put(8,4){\line(0,1){2}}
\put(8,9){\line(0,1){1}} \put(8,11){\line(0,1){1}}
\put(8,13){\line(0,1){1}} \put(8,17){\line(0,1){2}}
\put(8,22){\line(0,1){1}}

\put(9,1){\line(0,1){1}} \put(9,3){\line(0,1){1}}
\put(9,5){\line(0,1){1}} \put(9,9){\line(0,1){2}}
\put(9,14){\line(0,1){1}} \put(9,16){\line(0,1){1}}
\put(9,18){\line(0,1){1}} \put(9,22){\line(0,1){1}}

\put(10,2){\line(0,1){1}} \put(10,6){\line(0,1){1}}
\put(10,8){\line(0,1){1}} \put(10,10){\line(0,1){1}}
\put(10,14){\line(0,1){2}} \put(10,19){\line(0,1){1}}
\put(10,21){\line(0,1){1}}

\put(11,2){\line(0,1){1}} \put(11,6){\line(0,1){2}}
\put(11,11){\line(0,1){1}} \put(11,13){\line(0,1){1}}
\put(11,15){\line(0,1){1}} \put(11,19){\line(0,1){2}}

\put(12,3){\line(0,1){1}} \put(12,5){\line(0,1){1}}
\put(12,7){\line(0,1){1}} \put(12,11){\line(0,1){2}}
\put(12,16){\line(0,1){1}} \put(12,18){\line(0,1){1}}
\put(12,20){\line(0,1){1}}

\put(13,3){\line(0,1){2}} \put(13,8){\line(0,1){1}}
\put(13,10){\line(0,1){1}} \put(13,12){\line(0,1){1}}
\put(13,16){\line(0,1){2}} \put(13,21){\line(0,1){1}}

\put(14,2){\line(0,1){1}} \put(14,4){\line(0,1){1}}
\put(14,8){\line(0,1){2}} \put(14,13){\line(0,1){1}}
\put(14,15){\line(0,1){1}} \put(14,17){\line(0,1){1}}
\put(14,21){\line(0,1){1}}

\put(15,1){\line(0,1){1}} \put(15,5){\line(0,1){1}}
\put(15,7){\line(0,1){1}} \put(15,9){\line(0,1){1}}
\put(15,13){\line(0,1){2}} \put(15,18){\line(0,1){1}}
\put(15,20){\line(0,1){1}} \put(15,22){\line(0,1){1}}

\put(16,1){\line(0,1){1}} \put(16,5){\line(0,1){2}}
\put(16,10){\line(0,1){1}} \put(16,12){\line(0,1){1}}
\put(16,14){\line(0,1){1}} \put(16,18){\line(0,1){2}}
\put(16,23){\line(0,1){1}}

\put(17,2){\line(0,1){1}} \put(17,4){\line(0,1){1}}
\put(17,6){\line(0,1){1}} \put(17,10){\line(0,1){2}}
\put(17,15){\line(0,1){1}} \put(17,17){\line(0,1){1}}
\put(17,19){\line(0,1){1}} \put(17,23){\line(0,1){1}}

\put(18,3){\line(0,1){1}} \put(18,7){\line(0,1){1}}
\put(18,9){\line(0,1){1}} \put(18,11){\line(0,1){1}}
\put(18,15){\line(0,1){2}} \put(18,20){\line(0,1){1}}
\put(18,22){\line(0,1){1}}

\put(19,3){\line(0,1){1}} \put(19,7){\line(0,1){2}}
\put(19,12){\line(0,1){1}} \put(19,14){\line(0,1){1}}
\put(19,16){\line(0,1){1}} \put(19,20){\line(0,1){2}}

\put(20,4){\line(0,1){1}} \put(20,6){\line(0,1){1}}
\put(20,8){\line(0,1){1}} \put(20,12){\line(0,1){2}}
\put(20,17){\line(0,1){1}} \put(20,19){\line(0,1){1}}

\put(21,5){\line(0,1){1}} \put(21,9){\line(0,1){1}}
\put(21,11){\line(0,1){1}} \put(21,13){\line(0,1){1}}
\put(21,17){\line(0,1){2}}

\put(22,5){\line(0,1){1}} \put(22,9){\line(0,1){2}}
\put(22,14){\line(0,1){1}} \put(22,16){\line(0,1){1}}

\put(23,6){\line(0,1){1}} \put(23,8){\line(0,1){1}}
\put(23,15){\line(0,1){1}}

\put(24,7){\line(0,1){1}}

\put(7,0){\line(1,0){1}}

\put(6,1){\line(1,0){1}} \put(8,1){\line(1,0){1}}
\put(15,1){\line(1,0){1}}

\put(5,2){\line(1,0){1}} \put(9,2){\line(1,0){2}}
\put(14,2){\line(1,0){1}} \put(16,2){\line(1,0){1}}

\put(5,3){\line(1,0){1}} \put(9,3){\line(1,0){1}}
\put(11,3){\line(1,0){1}} \put(13,3){\line(1,0){1}}
\put(17,3){\line(1,0){2}}

\put(4,4){\line(1,0){1}} \put(6,4){\line(1,0){1}}
\put(8,4){\line(1,0){1}} \put(12,4){\line(1,0){2}}
\put(17,4){\line(1,0){1}} \put(19,4){\line(1,0){1}}

\put(3,5){\line(1,0){1}} \put(7,5){\line(1,0){2}}
\put(12,5){\line(1,0){1}} \put(14,5){\line(1,0){1}}
\put(16,5){\line(1,0){1}} \put(20,5){\line(1,0){2}}

\put(3,6){\line(1,0){1}} \put(7,6){\line(1,0){1}}
\put(9,6){\line(1,0){1}} \put(11,6){\line(1,0){1}}
\put(15,6){\line(1,0){2}} \put(20,6){\line(1,0){1}}
\put(22,6){\line(1,0){1}}

\put(2,7){\line(1,0){1}} \put(4,7){\line(1,0){1}}
\put(6,7){\line(1,0){1}} \put(10,7){\line(1,0){2}}
\put(15,7){\line(1,0){1}} \put(17,7){\line(1,0){1}}
\put(19,7){\line(1,0){1}} \put(23,7){\line(1,0){1}}

\put(1,8){\line(1,0){1}} \put(5,8){\line(1,0){2}}
\put(10,8){\line(1,0){1}} \put(12,8){\line(1,0){1}}
\put(14,8){\line(1,0){1}} \put(18,8){\line(1,0){2}}
\put(23,8){\line(1,0){1}}

\put(1,9){\line(1,0){1}} \put(5,9){\line(1,0){1}}
\put(7,9){\line(1,0){1}} \put(9,9){\line(1,0){1}}
\put(13,9){\line(1,0){2}} \put(18,9){\line(1,0){1}}
\put(20,9){\line(1,0){1}} \put(22,9){\line(1,0){1}}

\put(2,10){\line(1,0){1}} \put(4,10){\line(1,0){1}}
\put(8,10){\line(1,0){2}} \put(13,10){\line(1,0){1}}
\put(15,10){\line(1,0){1}} \put(17,10){\line(1,0){1}}
\put(21,10){\line(1,0){1}}

\put(3,11){\line(1,0){2}} \put(8,11){\line(1,0){1}}
\put(10,11){\line(1,0){1}} \put(12,11){\line(1,0){1}}
\put(16,11){\line(1,0){2}} \put(21,11){\line(1,0){1}}

\put(3,12){\line(1,0){1}} \put(5,12){\line(1,0){1}}
\put(7,12){\line(1,0){1}} \put(11,12){\line(1,0){2}}
\put(16,12){\line(1,0){1}} \put(18,12){\line(1,0){1}}
\put(20,12){\line(1,0){1}}

\put(2,13){\line(1,0){1}} \put(6,13){\line(1,0){2}}
\put(11,13){\line(1,0){1}} \put(13,13){\line(1,0){1}}
\put(15,13){\line(1,0){1}} \put(19,13){\line(1,0){2}}

\put(2,14){\line(1,0){1}} \put(6,14){\line(1,0){1}}
\put(8,14){\line(1,0){1}} \put(10,14){\line(1,0){1}}
\put(14,14){\line(1,0){2}} \put(19,14){\line(1,0){1}}
\put(21,14){\line(1,0){1}}

\put(1,15){\line(1,0){1}} \put(3,15){\line(1,0){1}}
\put(5,15){\line(1,0){1}} \put(9,15){\line(1,0){2}}
\put(14,15){\line(1,0){1}} \put(16,15){\line(1,0){1}}
\put(18,15){\line(1,0){1}} \put(22,15){\line(1,0){1}}

\put(0,16){\line(1,0){1}} \put(4,16){\line(1,0){2}}
\put(9,16){\line(1,0){1}} \put(11,16){\line(1,0){1}}
\put(13,16){\line(1,0){1}} \put(17,16){\line(1,0){2}}
\put(22,16){\line(1,0){1}}

\put(0,17){\line(1,0){1}} \put(4,17){\line(1,0){1}}
\put(6,17){\line(1,0){1}} \put(8,17){\line(1,0){1}}
\put(12,17){\line(1,0){2}} \put(17,17){\line(1,0){1}}
\put(19,17){\line(1,0){1}} \put(21,17){\line(1,0){1}}

\put(1,18){\line(1,0){1}} \put(3,18){\line(1,0){1}}
\put(7,18){\line(1,0){2}} \put(12,18){\line(1,0){1}}
\put(14,18){\line(1,0){1}} \put(16,18){\line(1,0){1}}
\put(20,18){\line(1,0){1}}

\put(2,19){\line(1,0){2}} \put(7,19){\line(1,0){1}}
\put(9,19){\line(1,0){1}} \put(11,19){\line(1,0){1}}
\put(15,19){\line(1,0){2}} \put(20,19){\line(1,0){1}}

\put(4,20){\line(1,0){1}} \put(6,20){\line(1,0){1}}
\put(10,20){\line(1,0){2}} \put(15,20){\line(1,0){1}}
\put(17,20){\line(1,0){1}} \put(19,20){\line(1,0){1}}

\put(5,21){\line(1,0){2}} \put(10,21){\line(1,0){1}}
\put(12,21){\line(1,0){1}} \put(14,21){\line(1,0){1}}
\put(18,21){\line(1,0){1}}

\put(7,22){\line(1,0){1}} \put(9,22){\line(1,0){1}}
\put(13,22){\line(1,0){2}} \put(18,22){\line(1,0){1}}

\put(8,23){\line(1,0){1}} \put(15,23){\line(1,0){1}}
\put(17,23){\line(1,0){1}}

\put(16,24){\line(1,0){1}}

\scriptsize{  \put(4.38,19.23){0} \put(5.38,19.23){2}
\put(4.38,18.23){3} \put(6.38,19.23){4} \put(5.38,18.23){5}
\put(4.38,17.23){6} \put(6.38,18.23){7} \put(5.38,17.23){8}
\put(7.18,18.23){9} \put(6.18,17.23){10} \put(5.18,16.23){11}

\put(7.38,19.23){6} \put(8.38,19.23){8} \put(9.18,19.23){10}
\put(8.18,18.23){11} \put(4.18,16.23){9}

\put(7.18,17.23){12} \put(6.18,16.23){13} \put(5.18,15.23){14}
\put(7.18,16.23){15} \put(6.18,15.23){16} \put(8.18,16.23){17}
\put(7.18,15.23){18} \put(6.18,14.23){19} \put(8.18,15.23){20}
\put(7.18,14.23){21} \put(9.18,15.23){22} \put(8.18,14.23){23}
\put(7.18,13.23){24}

\put(10.18,19.23){12} \put(9.18,18.23){13} \put(8.18,17.23){14}
\put(10.18,18.23){15} \put(9.18,17.23){16} \put(11.18,18.23){17}
\put(10.18,17.23){18} \put(9.18,16.23){19} \put(11.18,17.23){20}
\put(10.18,16.23){21} \put(12.18,17.23){22} \put(11.18,16.23){23}
\put(10.18,15.23){24}

\put(11.18,19.23){14} \put(12.18,19.23){16} \put(13.18,19.23){18}
\put(12.18,18.23){19} \put(14.18,19.23){20} \put(13.18,18.23){21}
\put(15.18,19.23){22} \put(14.18,18.23){23} \put(13.18,17.23){24}

\put(4.18,15.23){12} \put(4.18,14.23){15} \put(5.18,14.23){17}
\put(4.18,13.23){18} \put(5.18,13.23){20} \put(4.18,12.23){21}
\put(6.18,13.23){22} \put(5.18,12.23){23} \put(4.18,11.23){24}

\put(9.18,14.23){25} \put(8.18,13.23){26} \put(7.18,12.23){27}
\put(9.18,13.23){28} \put(8.18,12.23){29} \put(10.18,13.23){30}
\put(9.18,12.23){31} \put(8.18,11.23){32} \put(10.18,12.23){33}
\put(9.18,11.23){34} \put(11.18,12.23){35} \put(10.18,11.23){36}
\put(9.18,10.23){37}

\put(6.18,12.23){25} \put(5.18,11.23){26} \put(4.18,10.23){27}
\put(6.18,11.23){28} \put(5.18,10.23){29} \put(7.18,11.23){30}
\put(6.18,10.23){31} \put(5.18,9.23){32} \put(7.18,10.23){33}
\put(6.18,9.23){34} \put(8.18,10.23){35} \put(7.18,9.23){36}
\put(6.18,8.23){37}

\put(4.18,9.23){30} \put(4.18,8.23){33} \put(4.18,7.23){36}
\put(5.18,8.23){35}

\put(12.18,16.23){25} \put(11.18,15.23){26} \put(10.18,14.23){27}
\put(12.18,15.23){28} \put(11.18,14.23){29} \put(13.18,15.23){30}
\put(12.18,14.23){31} \put(11.18,13.23){32} \put(13.18,14.23){33}
\put(12.18,13.23){34} \put(14.18,14.23){35} \put(13.18,13.23){36}
\put(12.18,12.23){37}

\put(15.18,18.23){25} \put(14.18,17.23){26} \put(13.18,16.23){27}
\put(15.18,17.23){28} \put(14.18,16.23){29} \put(16.18,17.23){30}
\put(15.18,16.23){31} \put(14.18,15.23){32} \put(16.18,16.23){33}
\put(15.18,15.23){34} \put(17.18,16.23){35} \put(16.18,15.23){36}
\put(15.18,14.23){37}

\put(16.18,19.23){24}

\put(17.18,19.23){26} \put(16.18,18.23){27} \put(18.18,19.23){28}
\put(17.18,18.23){29} \put(19.18,19.23){30} \put(18.18,18.23){31}
\put(17.18,17.23){32} \put(19.18,18.23){33} \put(18.18,17.23){34}
\put(19.18,17.23){36} \put(18.18,16.23){37}

\put(11.18,11.23){38} \put(10.18,10.23){39} \put(9.18,9.23){40}
\put(11.18,10.23){41} \put(10.18,9.23){42} \put(12.18,10.23){43}
\put(11.18,9.23){44} \put(10.18,8.23){45} \put(12.18,9.23){46}
\put(11.18,8.23){47} \put(13.18,9.23){48} \put(12.18,8.23){49}
\put(11.18,7.23){50}

\put(8.18,9.23){38} \put(7.18,8.23){39} \put(6.18,7.23){40}
\put(8.18,8.23){41} \put(7.18,7.23){42} \put(9.18,8.23){43}
\put(8.18,7.23){44} \put(7.18,6.23){45} \put(9.18,7.23){46}
\put(8.18,6.23){47} \put(10.18,7.23){48} \put(9.18,6.23){49}
\put(8.18,5.23){50}

\put(5.18,7.23){38} \put(4.18,6.23){39} \put(5.18,6.23){41}
\put(4.18,5.23){42} \put(6.18,6.23){43} \put(5.18,5.23){44}
\put(4.18,4.23){45} \put(6.18,5.23){46} \put(5.18,4.23){47}
\put(7.18,5.23){48} \put(6.18,4.23){49}

\put(14.18,13.23){38} \put(13.18,12.23){39} \put(12.18,11.23){40}
\put(14.18,12.23){41} \put(13.18,11.23){42} \put(15.18,12.23){43}
\put(14.18,11.23){44} \put(13.18,10.23){45} \put(15.18,11.23){46}
\put(14.18,10.23){47} \put(16.18,11.23){48} \put(15.18,10.23){49}
\put(14.18,9.23){50}

\put(17.18,15.23){38} \put(16.18,14.23){39} \put(15.18,13.23){40}
\put(17.18,14.23){41} \put(16.18,13.23){42} \put(18.18,14.23){43}
\put(17.18,13.23){44} \put(16.18,12.23){45} \put(18.18,13.23){46}
\put(17.18,12.23){47} \put(19.18,13.23){48} \put(18.18,12.23){49}
\put(17.18,11.23){50}

\put(19.18,16.23){39} \put(18.18,15.23){40} \put(19.18,15.23){42}
\put(19.18,14.23){45}

\put(13.18,8.23){51} \put(12.18,7.23){52} \put(11.18,6.23){53}
\put(13.18,7.23){54} \put(12.18,6.23){55}  \put(14.18,7.23){56}
\put(13.18,6.23){57} \put(12.18,5.23){58} \put(14.18,6.23){59}
\put(13.18,5.23){60}  \put(15.18,6.23){61} \put(14.18,5.23){62}
\put(13.18,4.23){63}

\put(10.18,6.23){51} \put(9.18,5.23){52} \put(8.18,4.23){53}
\put(10.18,5.23){54} \put(9.18,4.23){55} \put(11.18,5.23){56}
\put(10.18,4.23){57} \put(11.18,4.23){59} \put(12.18,4.23){61}

\put(7.18,4.23){51}

\put(16.18,10.23){51} \put(15.18,9.23){52} \put(14.18,8.23){53}
\put(16.18,9.23){54} \put(15.18,8.23){55}  \put(17.18,9.23){56}
\put(16.18,8.23){57} \put(15.18,7.23){58} \put(17.18,8.23){59}
\put(16.18,7.23){60}  \put(18.18,8.23){61} \put(17.18,7.23){62}
\put(16.18,6.23){63}

\put(19.18,12.23){51} \put(18.18,11.23){52} \put(17.18,10.23){53}
\put(19.18,11.23){54} \put(18.18,10.23){55} \put(19.18,10.23){57}
\put(18.18,9.23){58} \put(19.18,9.23){60} \put(19.18,8.23){63}

\put(15.18,5.23){64} \put(14.18,4.23){65} \put(15.18,4.23){67}
\put(16.18,4.23){69}

\put(18.18,7.23){64} \put(17.18,6.23){65} \put(16.18,5.23){66}
\put(18.18,6.23){67} \put(17.18,5.23){68}  \put(19.18,6.23){69}
\put(18.18,5.23){70} \put(17.18,4.23){71} \put(19.18,5.23){72}
\put(18.18,4.23){73}  \put(19.18,4.23){75}

\put(19.18,7.23){66}}
\end{picture}}

\subfigure[]{ \setlength{\unitlength}{0.32cm}
\begin{picture}(25.5,24)(0,0)
\linethickness{0.01mm}
\multiput(4,4)(1,0){17}%
{\line(0,1){16}}
\multiput(4,4)(0,1){17}%
{\line(1,0){16}} \linethickness{0.6mm} \put(0,16){\line(0,1){1}}
\put(1,15){\line(0,1){1}} \put(1,8){\line(0,1){1}}
\put(1,17){\line(0,1){1}} \put(2,7){\line(0,1){1}}
\put(2,9){\line(0,1){1}} \put(2,13){\line(0,1){2}}
\put(2,18){\line(0,1){1}}

\put(3,5){\line(0,1){2}} \put(3,10){\line(0,1){1}}
\put(3,12){\line(0,1){1}} \put(3,14){\line(0,1){1}}
\put(3,18){\line(0,1){1}}

\put(4,4){\line(0,1){1}} \put(4,6){\line(0,1){1}}
\put(4,10){\line(0,1){2}} \put(4,15){\line(0,1){1}}
\put(4,17){\line(0,1){1}} \put(4,19){\line(0,1){1}}

\put(5,2){\line(0,1){2}} \put(5,7){\line(0,1){1}}
\put(5,9){\line(0,1){1}} \put(5,11){\line(0,1){1}}
\put(5,15){\line(0,1){2}} \put(5,20){\line(0,1){1}}

\put(6,1){\line(0,1){1}} \put(6,3){\line(0,1){1}}
\put(6,7){\line(0,1){2}} \put(6,12){\line(0,1){1}}
\put(6,14){\line(0,1){1}} \put(6,16){\line(0,1){1}}
\put(6,20){\line(0,1){1}}

\put(7,0){\line(0,1){1}} \put(7,4){\line(0,1){1}}
\put(7,6){\line(0,1){1}} \put(7,8){\line(0,1){1}}
\put(7,12){\line(0,1){2}} \put(7,17){\line(0,1){1}}
\put(7,19){\line(0,1){1}} \put(7,21){\line(0,1){1}}

\put(8,0){\line(0,1){1}} \put(8,4){\line(0,1){2}}
\put(8,9){\line(0,1){1}} \put(8,11){\line(0,1){1}}
\put(8,13){\line(0,1){1}} \put(8,17){\line(0,1){2}}
\put(8,22){\line(0,1){1}}

\put(9,1){\line(0,1){1}} \put(9,3){\line(0,1){1}}
\put(9,5){\line(0,1){1}} \put(9,9){\line(0,1){2}}
\put(9,14){\line(0,1){1}} \put(9,16){\line(0,1){1}}
\put(9,18){\line(0,1){1}} \put(9,22){\line(0,1){1}}

\put(10,2){\line(0,1){1}} \put(10,6){\line(0,1){1}}
\put(10,8){\line(0,1){1}} \put(10,10){\line(0,1){1}}
\put(10,14){\line(0,1){2}} \put(10,19){\line(0,1){1}}
\put(10,21){\line(0,1){1}}

\put(11,2){\line(0,1){1}} \put(11,6){\line(0,1){2}}
\put(11,11){\line(0,1){1}} \put(11,13){\line(0,1){1}}
\put(11,15){\line(0,1){1}} \put(11,19){\line(0,1){2}}

\put(12,3){\line(0,1){1}} \put(12,5){\line(0,1){1}}
\put(12,7){\line(0,1){1}} \put(12,11){\line(0,1){2}}
\put(12,16){\line(0,1){1}} \put(12,18){\line(0,1){1}}
\put(12,20){\line(0,1){1}}

\put(13,3){\line(0,1){2}} \put(13,8){\line(0,1){1}}
\put(13,10){\line(0,1){1}} \put(13,12){\line(0,1){1}}
\put(13,16){\line(0,1){2}} \put(13,21){\line(0,1){1}}

\put(14,2){\line(0,1){1}} \put(14,4){\line(0,1){1}}
\put(14,8){\line(0,1){2}} \put(14,13){\line(0,1){1}}
\put(14,15){\line(0,1){1}} \put(14,17){\line(0,1){1}}
\put(14,21){\line(0,1){1}}

\put(15,1){\line(0,1){1}} \put(15,5){\line(0,1){1}}
\put(15,7){\line(0,1){1}} \put(15,9){\line(0,1){1}}
\put(15,13){\line(0,1){2}} \put(15,18){\line(0,1){1}}
\put(15,20){\line(0,1){1}} \put(15,22){\line(0,1){1}}

\put(16,1){\line(0,1){1}} \put(16,5){\line(0,1){2}}
\put(16,10){\line(0,1){1}} \put(16,12){\line(0,1){1}}
\put(16,14){\line(0,1){1}} \put(16,18){\line(0,1){2}}
\put(16,23){\line(0,1){1}}

\put(17,2){\line(0,1){1}} \put(17,4){\line(0,1){1}}
\put(17,6){\line(0,1){1}} \put(17,10){\line(0,1){2}}
\put(17,15){\line(0,1){1}} \put(17,17){\line(0,1){1}}
\put(17,19){\line(0,1){1}} \put(17,23){\line(0,1){1}}

\put(18,3){\line(0,1){1}} \put(18,7){\line(0,1){1}}
\put(18,9){\line(0,1){1}} \put(18,11){\line(0,1){1}}
\put(18,15){\line(0,1){2}} \put(18,20){\line(0,1){1}}
\put(18,22){\line(0,1){1}}

\put(19,3){\line(0,1){1}} \put(19,7){\line(0,1){2}}
\put(19,12){\line(0,1){1}} \put(19,14){\line(0,1){1}}
\put(19,16){\line(0,1){1}} \put(19,20){\line(0,1){2}}

\put(20,4){\line(0,1){1}} \put(20,6){\line(0,1){1}}
\put(20,8){\line(0,1){1}} \put(20,12){\line(0,1){2}}
\put(20,17){\line(0,1){1}} \put(20,19){\line(0,1){1}}

\put(21,5){\line(0,1){1}} \put(21,9){\line(0,1){1}}
\put(21,11){\line(0,1){1}} \put(21,13){\line(0,1){1}}
\put(21,17){\line(0,1){2}}

\put(22,5){\line(0,1){1}} \put(22,9){\line(0,1){2}}
\put(22,14){\line(0,1){1}} \put(22,16){\line(0,1){1}}

\put(23,6){\line(0,1){1}} \put(23,8){\line(0,1){1}}
\put(23,15){\line(0,1){1}}

\put(24,7){\line(0,1){1}}

\put(7,0){\line(1,0){1}}

\put(6,1){\line(1,0){1}} \put(8,1){\line(1,0){1}}
\put(15,1){\line(1,0){1}}

\put(5,2){\line(1,0){1}} \put(9,2){\line(1,0){2}}
\put(14,2){\line(1,0){1}} \put(16,2){\line(1,0){1}}

\put(5,3){\line(1,0){1}} \put(9,3){\line(1,0){1}}
\put(11,3){\line(1,0){1}} \put(13,3){\line(1,0){1}}
\put(17,3){\line(1,0){2}}

\put(4,4){\line(1,0){1}} \put(6,4){\line(1,0){1}}
\put(8,4){\line(1,0){1}} \put(12,4){\line(1,0){2}}
\put(17,4){\line(1,0){1}} \put(19,4){\line(1,0){1}}

\put(3,5){\line(1,0){1}} \put(7,5){\line(1,0){2}}
\put(12,5){\line(1,0){1}} \put(14,5){\line(1,0){1}}
\put(16,5){\line(1,0){1}} \put(20,5){\line(1,0){2}}

\put(3,6){\line(1,0){1}} \put(7,6){\line(1,0){1}}
\put(9,6){\line(1,0){1}} \put(11,6){\line(1,0){1}}
\put(15,6){\line(1,0){2}} \put(20,6){\line(1,0){1}}
\put(22,6){\line(1,0){1}}

\put(2,7){\line(1,0){1}} \put(4,7){\line(1,0){1}}
\put(6,7){\line(1,0){1}} \put(10,7){\line(1,0){2}}
\put(15,7){\line(1,0){1}} \put(17,7){\line(1,0){1}}
\put(19,7){\line(1,0){1}} \put(23,7){\line(1,0){1}}

\put(1,8){\line(1,0){1}} \put(5,8){\line(1,0){2}}
\put(10,8){\line(1,0){1}} \put(12,8){\line(1,0){1}}
\put(14,8){\line(1,0){1}} \put(18,8){\line(1,0){2}}
\put(23,8){\line(1,0){1}}

\put(1,9){\line(1,0){1}} \put(5,9){\line(1,0){1}}
\put(7,9){\line(1,0){1}} \put(9,9){\line(1,0){1}}
\put(13,9){\line(1,0){2}} \put(18,9){\line(1,0){1}}
\put(20,9){\line(1,0){1}} \put(22,9){\line(1,0){1}}

\put(2,10){\line(1,0){1}} \put(4,10){\line(1,0){1}}
\put(8,10){\line(1,0){2}} \put(13,10){\line(1,0){1}}
\put(15,10){\line(1,0){1}} \put(17,10){\line(1,0){1}}
\put(21,10){\line(1,0){1}}

\put(3,11){\line(1,0){2}} \put(8,11){\line(1,0){1}}
\put(10,11){\line(1,0){1}} \put(12,11){\line(1,0){1}}
\put(16,11){\line(1,0){2}} \put(21,11){\line(1,0){1}}

\put(3,12){\line(1,0){1}} \put(5,12){\line(1,0){1}}
\put(7,12){\line(1,0){1}} \put(11,12){\line(1,0){2}}
\put(16,12){\line(1,0){1}} \put(18,12){\line(1,0){1}}
\put(20,12){\line(1,0){1}}

\put(2,13){\line(1,0){1}} \put(6,13){\line(1,0){2}}
\put(11,13){\line(1,0){1}} \put(13,13){\line(1,0){1}}
\put(15,13){\line(1,0){1}} \put(19,13){\line(1,0){2}}

\put(2,14){\line(1,0){1}} \put(6,14){\line(1,0){1}}
\put(8,14){\line(1,0){1}} \put(10,14){\line(1,0){1}}
\put(14,14){\line(1,0){2}} \put(19,14){\line(1,0){1}}
\put(21,14){\line(1,0){1}}

\put(1,15){\line(1,0){1}} \put(3,15){\line(1,0){1}}
\put(5,15){\line(1,0){1}} \put(9,15){\line(1,0){2}}
\put(14,15){\line(1,0){1}} \put(16,15){\line(1,0){1}}
\put(18,15){\line(1,0){1}} \put(22,15){\line(1,0){1}}

\put(0,16){\line(1,0){1}} \put(4,16){\line(1,0){2}}
\put(9,16){\line(1,0){1}} \put(11,16){\line(1,0){1}}
\put(13,16){\line(1,0){1}} \put(17,16){\line(1,0){2}}
\put(22,16){\line(1,0){1}}

\put(0,17){\line(1,0){1}} \put(4,17){\line(1,0){1}}
\put(6,17){\line(1,0){1}} \put(8,17){\line(1,0){1}}
\put(12,17){\line(1,0){2}} \put(17,17){\line(1,0){1}}
\put(19,17){\line(1,0){1}} \put(21,17){\line(1,0){1}}

\put(1,18){\line(1,0){1}} \put(3,18){\line(1,0){1}}
\put(7,18){\line(1,0){2}} \put(12,18){\line(1,0){1}}
\put(14,18){\line(1,0){1}} \put(16,18){\line(1,0){1}}
\put(20,18){\line(1,0){1}}

\put(2,19){\line(1,0){2}} \put(7,19){\line(1,0){1}}
\put(9,19){\line(1,0){1}} \put(11,19){\line(1,0){1}}
\put(15,19){\line(1,0){2}} \put(20,19){\line(1,0){1}}

\put(4,20){\line(1,0){1}} \put(6,20){\line(1,0){1}}
\put(10,20){\line(1,0){2}} \put(15,20){\line(1,0){1}}
\put(17,20){\line(1,0){1}} \put(19,20){\line(1,0){1}}

\put(5,21){\line(1,0){2}} \put(10,21){\line(1,0){1}}
\put(12,21){\line(1,0){1}} \put(14,21){\line(1,0){1}}
\put(18,21){\line(1,0){1}}

\put(7,22){\line(1,0){1}} \put(9,22){\line(1,0){1}}
\put(13,22){\line(1,0){2}} \put(18,22){\line(1,0){1}}

\put(8,23){\line(1,0){1}} \put(15,23){\line(1,0){1}}
\put(17,23){\line(1,0){1}}

\put(16,24){\line(1,0){1}}

\scriptsize{  \put(4.18,19.23){39} \put(5.18,19.23){42}
\put(4.18,18.23){37} \put(6.18,19.23){45} \put(5.18,18.23){40}
\put(4.18,17.23){35} \put(6.18,18.23){43} \put(5.18,17.23){38}
\put(7.18,18.23){46} \put(6.18,17.23){41} \put(5.18,16.23){36}

\put(7.18,19.23){48} \put(8.18,19.23){51} \put(9.18,19.23){54}
\put(8.18,18.23){49} \put(4.18,16.23){33}

\put(7.18,17.23){44} \put(6.18,16.23){39} \put(5.18,15.23){34}
\put(7.18,16.23){42} \put(6.18,15.23){37} \put(8.18,16.23){45}
\put(7.18,15.23){40} \put(6.18,14.23){35} \put(8.18,15.23){43}
\put(7.18,14.23){38} \put(9.18,15.23){46} \put(8.18,14.23){41}
\put(7.18,13.23){36}

\put(10.18,19.23){57} \put(9.18,18.23){52} \put(8.18,17.23){47}
\put(10.18,18.23){55} \put(9.18,17.23){50} \put(11.18,18.23){58}
\put(10.18,17.23){53} \put(9.18,16.23){48} \put(11.18,17.23){56}
\put(10.18,16.23){51} \put(12.18,17.23){59} \put(11.18,16.23){54}
\put(10.18,15.23){49}

\put(11.18,19.23){60} \put(12.18,19.23){63} \put(13.18,19.23){66}
\put(12.18,18.23){61} \put(14.18,19.23){69} \put(13.18,18.23){64}
\put(15.18,19.23){72} \put(14.18,18.23){67} \put(13.18,17.23){62}

\put(4.18,15.23){31} \put(4.18,14.23){29} \put(5.18,14.23){32}
\put(4.18,13.23){27} \put(5.18,13.23){30} \put(4.18,12.23){25}
\put(6.18,13.23){33} \put(5.18,12.23){28} \put(4.18,11.23){23}

\put(9.18,14.23){44} \put(8.18,13.23){39} \put(7.18,12.23){34}
\put(9.18,13.23){42} \put(8.18,12.23){37} \put(10.18,13.23){45}
\put(9.18,12.23){40} \put(8.18,11.23){35} \put(10.18,12.23){43}
\put(9.18,11.23){38} \put(11.18,12.23){46} \put(10.18,11.23){41}
\put(9.18,10.23){36}

\put(6.18,12.23){31} \put(5.18,11.23){26} \put(4.18,10.23){21}
\put(6.18,11.23){29} \put(5.18,10.23){24} \put(7.18,11.23){32}
\put(6.18,10.23){27} \put(5.18,9.23){22} \put(7.18,10.23){30}
\put(6.18,9.23){25} \put(8.18,10.23){33} \put(7.18,9.23){28}
\put(6.18,8.23){23}

\put(4.18,9.23){19} \put(4.18,8.23){17} \put(4.18,7.23){15}
\put(5.18,8.23){20}

\put(12.18,16.23){57} \put(11.18,15.23){52} \put(10.18,14.23){47}
\put(12.18,15.23){55} \put(11.18,14.23){50} \put(13.18,15.23){58}
\put(12.18,14.23){53} \put(11.18,13.23){48} \put(13.18,14.23){56}
\put(12.18,13.23){51} \put(14.18,14.23){59} \put(13.18,13.23){54}
\put(12.18,12.23){49}

\put(15.18,18.23){70} \put(14.18,17.23){65} \put(13.18,16.23){60}
\put(15.18,17.23){68} \put(14.18,16.23){63} \put(16.18,17.23){71}
\put(15.18,16.23){66} \put(14.18,15.23){61} \put(16.18,16.23){69}
\put(15.18,15.23){64} \put(17.18,16.23){72} \put(16.18,15.23){67}
\put(15.18,14.23){62}

\put(16.18,19.23){75}

\put(17.18,19.23){78} \put(16.18,18.23){73} \put(18.18,19.23){81}
\put(17.18,18.23){76} \put(19.18,19.23){84} \put(18.18,18.23){79}
\put(17.18,17.23){74} \put(19.18,18.23){82} \put(18.18,17.23){77}
\put(19.18,17.23){80} \put(18.18,16.23){75}

\put(11.18,11.23){44} \put(10.18,10.23){39} \put(9.18,9.23){34}
\put(11.18,10.23){42} \put(10.18,9.23){37} \put(12.18,10.23){45}
\put(11.18,9.23){40} \put(10.18,8.23){35} \put(12.18,9.23){43}
\put(11.18,8.23){38} \put(13.18,9.23){46} \put(12.18,8.23){41}
\put(11.18,7.23){36}

\put(8.18,9.23){31} \put(7.18,8.23){26} \put(6.18,7.23){21}
\put(8.18,8.23){29} \put(7.18,7.23){24} \put(9.18,8.23){32}
\put(8.18,7.23){27} \put(7.18,6.23){22} \put(9.18,7.23){30}
\put(8.18,6.23){25} \put(10.18,7.23){33} \put(9.18,6.23){28}
\put(8.18,5.23){23}

\put(5.18,7.23){18} \put(4.18,6.23){13} \put(5.18,6.23){16}
\put(4.18,5.23){11} \put(6.18,6.23){19} \put(5.18,5.23){14}
\put(4.38,4.23){9} \put(6.18,5.23){17} \put(5.18,4.23){12}
\put(7.18,5.23){20} \put(6.18,4.23){15}

\put(14.18,13.23){57} \put(13.18,12.23){52} \put(12.18,11.23){47}
\put(14.18,12.23){55} \put(13.18,11.23){50} \put(15.18,12.23){58}
\put(14.18,11.23){53} \put(13.18,10.23){48} \put(15.18,11.23){56}
\put(14.18,10.23){51} \put(16.18,11.23){59} \put(15.18,10.23){54}
\put(14.18,9.23){49}

\put(17.18,15.23){70} \put(16.18,14.23){65} \put(15.18,13.23){60}
\put(17.18,14.23){68} \put(16.18,13.23){63} \put(18.18,14.23){71}
\put(17.18,13.23){66} \put(16.18,12.23){61} \put(18.18,13.23){69}
\put(17.18,12.23){64} \put(19.18,13.23){72} \put(18.18,12.23){67}
\put(17.18,11.23){62}

\put(19.18,16.23){78} \put(18.18,15.23){73} \put(19.18,15.23){76}
\put(19.18,14.23){74}

\put(13.18,8.23){44} \put(12.18,7.23){39} \put(11.18,6.23){34}
\put(13.18,7.23){42} \put(12.18,6.23){37}  \put(14.18,7.23){45}
\put(13.18,6.23){40} \put(12.18,5.23){35} \put(14.18,6.23){43}
\put(13.18,5.23){38}  \put(15.18,6.23){46} \put(14.18,5.23){41}
\put(13.18,4.23){36}

\put(10.18,6.23){31} \put(9.18,5.23){26} \put(8.18,4.23){21}
\put(10.18,5.23){29} \put(9.18,4.23){24} \put(11.18,5.23){32}
\put(10.18,4.23){27} \put(11.18,4.23){30} \put(12.18,4.23){33}

\put(7.18,4.23){18}

\put(16.18,10.23){57} \put(15.18,9.23){52} \put(14.18,8.23){47}
\put(16.18,9.23){55} \put(15.18,8.23){50}  \put(17.18,9.23){58}
\put(16.18,8.23){53} \put(15.18,7.23){48} \put(17.18,8.23){56}
\put(16.18,7.23){51}  \put(18.18,8.23){59} \put(17.18,7.23){54}
\put(16.18,6.23){49}

\put(19.18,12.23){70} \put(18.18,11.23){65} \put(17.18,10.23){60}
\put(19.18,11.23){68} \put(18.18,10.23){63} \put(19.18,10.23){55}
\put(18.18,9.23){61} \put(19.18,9.23){64} \put(19.18,8.23){62}

\put(15.18,5.23){44} \put(14.18,4.23){39} \put(15.18,4.23){42}
\put(16.18,4.23){45}

\put(18.18,7.23){57} \put(17.18,6.23){52} \put(16.18,5.23){47}
\put(18.18,6.23){55} \put(17.18,5.23){50}  \put(19.18,6.23){58}
\put(18.18,5.23){53} \put(17.18,4.23){48} \put(19.18,5.23){56}
\put(18.18,4.23){51}  \put(19.18,4.23){54}

\put(19.18,7.23){60}}
\end{picture}}
} \caption{(a) The coloring $\Psi_1$, and (b) the coloring
$\Psi_2$,  with $R=2$. The left upper corner coordinate is
$(-6,9)$.} \label{tiling}
\end{center}
\end{figure*}
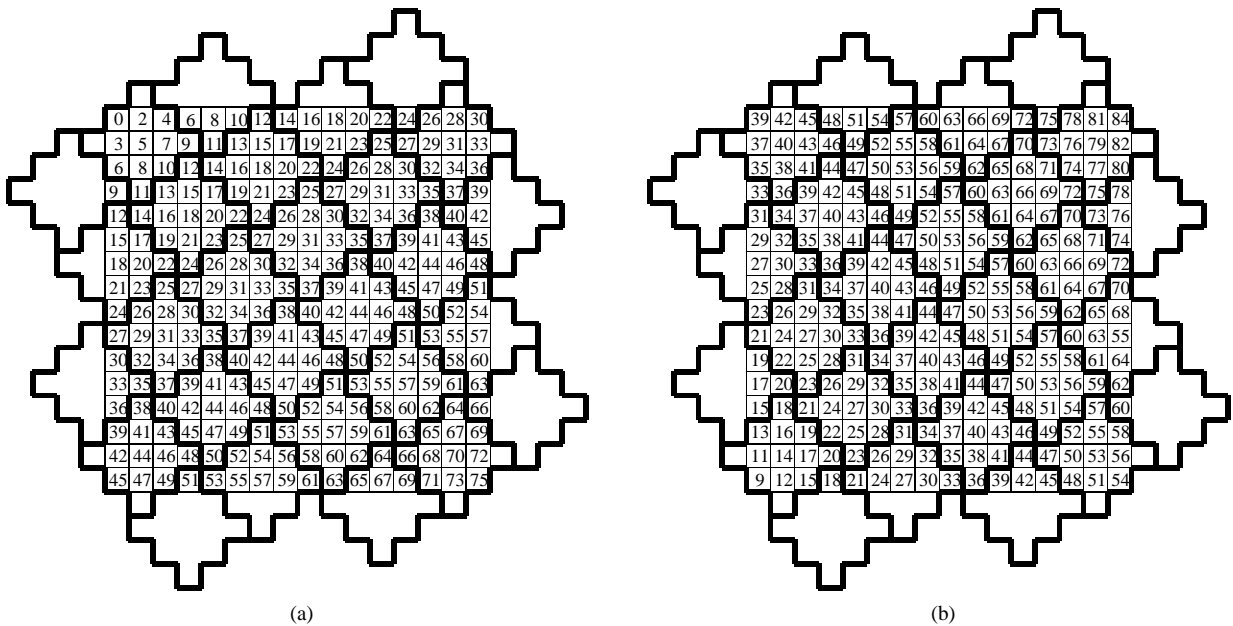

\subsection{Multidimensional codes}
\label{sec:multi_Lee}

A similar idea can be used for $D$-dimensional code of size $n_1
\times n_2 \times \cdots \times n_D$ correcting Lee sphere error
with radius one. We use $D$ different colorings of the array. For
a position $(i_1 , i_2 , \ldots , i_D )$, $0 \leq i_\ell \leq
n-1$, $1 \leq \ell \leq D$, the $s$-th coloring, $1 \leq s \leq
D$, assigns the color $\Sigma_{\ell=1}^D \ell \cdot i_{s+\ell-1}$,
where indices are residues modulo $D$ between 1 and $D$ (note that
each color in each coloring forms a perfect code~\cite{GoWe} when
we consider the coloring in $\Z^D$ and not just in the array).
Again, for each coloring, all the color numbers located in a Lee
sphere of the $D$-dimensional array are $2D+1$ consecutive
integers. Bit $k$ of the $s$-th component code is the binary sum
of all the bits colored with the integer $k$, by the $s$-th
coloring, in the $D$-dimensional codeword. It is easy to prove
that the coloring matrix is invertible. But, property (p.3) does
not hold in all dimensions. Therefore, if property (p.3) does not
hold for the $s$-th coloring we need to take a
$(2D+1)$-burst-correcting code instead of $(2D+1)$-burst-locator
code for the $s$-th dimension. The consequence will be a code with
larger redundancy. Clearly, there are many appropriate colorings
for each dimension. Hence, we can try to replace each coloring by
one for which property (p.3) holds. Of course we have to make sure
that the coloring matrix will be invertible. The excess redundancy
in this case is quadratic in $D$, compared to exponential in $D$
if we use the transformation $T$ and the code which corrects
$D$-dimensional box-error.

\section{Bursts with Limited Weight}
\label{sec:limitweight}

We now turn for a new kind of errors which are important in our
final goal of correcting an arbitrary cluster-error. Assume we
need to correct errors in a cluster of size $b$, where the number
of erroneous positions is at most $t$. We wish to find
one-dimensional and multidimensional codes which correct such
errors. These codes have an obvious application as we can expect
that in area that suffers from an event which caused errors, not
all positions were affected. Another important observation is that
any arbitrary cluster-error of size $b$ is located inside a Lee
sphere with radius $\lfloor \frac{b}{2} \rfloor$. Hence, this
cluster-error can be corrected if we can correct a Lee sphere
error with radius $\lfloor \frac{b}{2} \rfloor$, where the number
of erroneous positions is at most $b$.

\subsection{One-dimensional codes}
\label{sec:limited_one}

Throughout this subsection all codes are binary and $b$ is an odd
integer.

Let $\cC_1$ be a $t$-error-correcting code of length $b$, and
redundancy $r$. Let $H_1$ be its parity-check matrix of size
$r\times b$,
$$H_1=[h_0^1,h_1^1,\ldots,h_{b-1}^1].$$
Let $\cC_2$ be a $b$-burst-locator code, of length $n=2^m-1$, and
redundancy $m$. Its parity-check matrix, $H_2$, is of size
$m\times (2^m-1)$,
$$H_2=[h_0^2,h_1^2,\ldots,h_{n-1}^2].$$

Based on these two matrices, we construct a new parity-check
matrix $\cH$,
$$\cH=\left[ \begin{array}{ccccccc}
h_0^1 & h_1^1 & \cdots & h_{b-1}^1 & h_0^1 & \cdots &
h_{(n-1)(\text{mod} \ b)}^1 \\
h_0^2 & h_1^2 & \cdots & h_{b-1}^2 & h_b^2 & \cdots & h_{n-1}^2
\end{array}\right].$$
The $j$-th column, $0\leq j \leq n-1$, of the matrix $\cH$ will be
defined as the concatenation of the $j (\text{mod}\ b)$-th column
of the matrix $H_1$ and the $j$-th column of the matrix $H_2$,

\begin{lemma}
\label{lem:limit_burst} $\cH$ is a parity-check matrix for a code
$\cC$ of length $2^m-1$, correcting every burst of length $b$ with
at most $t$ erroneous positions.
\end{lemma}
\begin{proof}
It is sufficient to prove that there are no two bursts of length
$b$ and weight at most $t$ that have the same syndrome. Assume the
contrary, that there are two words, of weight at most $t$, with
bursts of length $b$, $y^1 = (y_0^1 , \ldots , y_{n-1}^1)$ and
$y^2 = (y_0^2 , \ldots , y_{n-1}^2)$, which have the same
syndrome, i.e. $\cH(y^1)^T = \cH(y^2)^T$. This implies that

\begin{align}
\label{eq:part1} \sum_{i=0}^{n-1}y_i^1h_{[i]_b}^1=
\sum_{i=0}^{n-1}y_i^2h_{[i]_b}^1~,
\end{align}
\begin{align}
\label{eq:part2} \sum_{i=0}^{n-1}y_i^1h_i^2=
\sum_{i=0}^{n-1}y_i^2h_i^2~,
\end{align}
where $[j]_\ell$ is the unique residue of $j (\text{mod}~\ell)$
between 0 and $\ell-1$. Let $i_0^1$, $i_0^2$ be the first nonzero
bit in $y^1$, $y^2$, respectively. Therefore, (\ref{eq:part1}) can
be written as
\begin{align*}
\sum_{i=i_0^1}^{i_0^1+b-1}y_i^1h_{[i]_b}^1=
\sum_{i=i_0^2}^{i_0^2+b-1}y_i^2h_{[i]_b}^1.
\end{align*}
If we denote $\ell^1=b\left\lfloor\frac{i_0^1}{b}\right\rfloor,
\ell^2=b\left\lfloor\frac{i_0^2}{b}\right\rfloor$ then the last
equation can be written in the following way:
\begin{eqnarray}
& & \sum_{i=0}^{[i_0^1]_b-1}y_{\ell^1+b+i}^1h_{i}^1 +
\sum_{i=[i_0^1]_b}^{b-1}y_{\ell^1+i}^1h_{i}^1=
\nonumber\\
& & \sum_{i=0}^{[i_0^2]_b-1}y_{\ell^2+b+i}^2h_{i}^1+
\sum_{i=[i_0^2]_b}^{b-1}y_{\ell^2+i}^2h_{i}^1.\nonumber
\end{eqnarray}
This equation implies that for the words of length $b$
\begin{eqnarray}
z^1 & = &
\left(y_{\ell^1+b}^1,\ldots,y_{\ell^1+b+[i_0^1]_b-1}^1,y_{\ell^1+[i_0^1]_b}^1,
\ldots,y_{\ell^1+b-1}^1\right), \nonumber \\
z^2 & = &
\left(y_{\ell^2+b}^2,\ldots,y_{\ell^2+b+[i_0^2]_b-1}^2,y_{\ell^2+[i_0^2]_b}^2,
\ldots,y_{\ell^2+b-1}^2\right), \nonumber
\end{eqnarray}
we have $H_1 (z^1)^T = H_1 (z^2)^T$. The weight of these words is
at most $t$ and since $H_1$ is a parity-check matrix of a
$t$-error-correcting code we have that $z^1=z^2$. Therefore, the
two different words $y^1$ and $y^2$ contain the same burst of
length $b$ up to a cyclic permutation.
By Lemma~\ref{lem:locator} we have that $H_2 (y^1)^T \neq H_2
(y^2)^T$, contradicting~(\ref{eq:part2}).

Thus, $\cH$ is a parity-check matrix for a code of length $2^m-1$,
correcting every burst of length $b$ with at most $t$ erroneous
positions.
\end{proof}
\noindent {\bf Remark:} It is important to note that $\cH$ is
acting similarly to the constructions of previous sections. The
first part which is a concatenation of several copies of $H_1$ is
a parity-check matrix of a $b$-burst-correcting code with at most
$t$ erroneous positions. It finds the burst pattern up to a cyclic
shift. The second part of $\cH$, $H_2$, is a parity-check matrix
of a $b$-burst-locator which can find the location of a burst of
length $b$ given up to a cyclic shift. Lemma~\ref{lem:limit_burst}
gives a formal proof for these facts in terms of the generated
syndromes.

To summarize the parameters of the construction we need the
parameters of $t$-error-correcting codes of length $b$. We can use
BCH codes~\cite{McSl,Rot} for this purpose. If $\ell$ is the least
integer such that $b \leq 2^\ell-1$ then there exists a
$t$-error-correcting BCH code of length $2^\ell-1$ and redundancy
at most $t\ell$. By shortening we can obtain a
$t$-error-correcting code of length $b$ and redundancy $t \lceil
\text{log}_2 b \rceil$. Now, we can summarize our construction in
this section.

\begin{theorem}
\label{thm:limited_burst} The code $\cC$ has length $n=2^m-1$ and
it corrects every burst of length $b$ with at most $t$ erroneous
positions. If $b$ is an odd integer then the redundancy of the
code is $m+t \lceil \text{log}_2 b \rceil \leq \lceil \text{log}_2
n \rceil +t \lceil \text{log}_2 b \rceil$, and if $b$ is an odd
integer then the redundancy of the code is $m+t \lceil
\text{log}_2 (b+1) \rceil \leq \lceil \text{log}_2 n \rceil +t
\lceil \text{log}_2 (b+1) \rceil$.
\end{theorem}

\subsection{Multidimensional codes}

Now, we want to design a two-dimensional code of size $n_1 \times
n_2$ capable of correcting a $(b_1 \times b_2 )$-cluster with
weight at most $t$, where $b_1 b_2$ is an odd integer. We use a
construction similar to the one used in previous sections. In this
construction we have to use two binary component codes. The
vertical one is a $(b_1 b_2 )$-burst-correcting code, of length
$n_1 b_2$ in which the weight of the burst is at most $t$. Such a
code was constructed in the previous subsection. The horizontal
component code is a $(b_1 b_2 )$-burst-locator code of length $n_2
b_1$. They are used in the same manner as they are used in
previous sections to correct a burst of size $b_1 \times b_2$.
Since the vertical code can find a burst only if the weight of it
is at most $t$, it follows that the two-dimensional code can
handle bursts of size $b_1 \times b_2$ only if their weight is at
most $t$. If $b_1b_2$ is an odd integer than by
Theorem~\ref{thm:limited_burst} we can take a vertical code of
length $n_1 b_2$, $2^m > n_1 b_2 \geq 2^{m-1}$, with at most $m+t
\lceil \text{log}_2 (b_1b_2) \rceil$ redundancy bits. The
horizontal code has length $2^{r-b_1b_2+1}-1$, $2^{r-b_1b_2+1}
> n_2 b_1 \geq 2^{r-b_1b_2}$, with $r-b_1b_2+1$ redundancy bits.
Therefore we have.

\begin{theorem}
\label{thm:2D_limited} The redundancy of the $n_1 \times n_2$ code
$\cC$ which is capable to correct a $(b_1 \times b_2)$-cluster,
$b_1b_2$ an odd integer, with weight less or equal $t$, is at most
$\lceil \text{log}_2 (n_1n_2) \rceil +(t+1) \lceil \text{log}_2
(b_1 b_2) \rceil +3$.
\end{theorem}

The generalization for multidimensional codes is straight forward.
If the size of the cluster is an even integer then we will use an
appropriate coloring.

\section{Correction of Arbitrary Bursts}
\label{sec:arbitrary}

Finally, we want to design a code which corrects an arbitrary
$D$-dimensional cluster-error of size $b$. If $b$ is odd then the
cluster is located inside a Lee sphere with radius
$\frac{b-1}{2}$. If $b$ is even then either we consider it as a
cluster-error of size $b+1$ or slightly modify the constructions
for a small improvement on the efficiency. Modification based on
coloring can be also obtained if $b=3$ and $D \geq 3$.

We now show how the codes of Section~\ref{sec:limitweight} help to
correct an arbitrary cluster-error of size $b$. We start again
with two-dimensional codes. A cluster-error of size $b$ is located
inside a $b \times b$ square. Therefore, we can use a code of size
$n_1 \times n_2$ which corrects a $(b \times b)$-cluster with
weight $b$. The generalization for $D$-dimensional code is
straight forward.

An improvement in the excess redundancy is obtained if we consider
a code which corrects a smaller shape, with limited weight, in
which the $b$-cluster is located. For simplicity we will consider
only the case where $b=2R+1$ and the small shape is a Lee sphere
with radius $R$. We will use the colorings $\Psi_1$ and $\Psi_2$
given in Section~\ref{sec:Lee} and two component codes, the first
one is a $b^*$-burst-correcting code, $b^*=2R^2+2R+1$, in which
the weight of the burst is at most $b$, and the second one is
$b^*$-burst-locator code. Now, we apply the coloring method of
Section~\ref{sec:coloring}. The redundancy computations are
similar to the ones in subsection~\ref{subsec:tiling} and in
Theorem~\ref{thm:2D_limited} and we have the following result.

\begin{theorem}
\label{thm:upper_excess} the redundancy of an $n_1 \times n_2$
code capable to correct a cluster of size $b$ is at most $\lceil
\text{log}_2 (n_1n_2) \rceil +(b+1) \lceil 2 \text{log}_2 b \rceil
+3$.
\end{theorem}

Generalization for $D$-dimensional code is done by using the
transformation $T$ of Section~\ref{sec:Lee}. Hence, there is some
loose of efficiency, but the performance is still better than the
performance of a code which corrects a $D$-dimensional cluster
error whose shape is a $D$-dimensional box-error with limited
weight $b$. The redundancy computation is similar to the ones in
previous sections.

The next question of interest is a lower bound on the excess
redundancy of a code which corrects an arbitrary cluster of size
$b$. A lower bound on the excess redundancy is $\text{log}_2
N_D(b)$, where $N_D(b)$ is the number of distinct patterns
considered as $D$-dimensional clusters of size $b$. Finding bounds
on $N_D(b)$ is an interesting geometrical combinatorial problems
of itself. A related question is to find number of distinct
clusters with size $b$, with no "holes", and exactly $b$ erroneous
positions. This problem is the same as finding the number of
$b$-polyominos. For $D=2$ the known lower bound on their number is
$3.981037^b$~\cite{BMRR} and the known upper bound is
$4.649551^b$~\cite{KlRi}. Therefore, we have

\begin{theorem}
\label{thm:lower_excess} The excess redundancy of a
two-dimensional code, which is capable to correct an arbitrary
cluster of size $b$, is at least $b \cdot \text{log}_2 3.981037~$.
\end{theorem}

Theorem~\ref{thm:lower_excess} is a small improvement of the
trivial lower bound (which is $b$) on the excess redundancy. But,
the gap between the orders of the lower bound $O(b)$ and the upper
bound $O(b \cdot \text{log}_2b)$ is still large.

\section{Representation with Parity-Check Matrices}
\label{sec:parity-check}

All codes which were discussed in the previous sections are using
auxiliary linear codes for the computation of the redundancy bits
in the codewords and to conduct the proper decoding. It is not
difficult to see that all the multidimensional codes are linear,
by noting that the bit by bit addition of two codewords is another
codeword. In this section we will explain how to present similar
codes with parity-check matrices. This will be done by considering
all the component codes as binary codes.\\
{\bf Remark:} When the component code is a linear burst-correcting
code of length $n$ over GF($2^b$) we can consider it as a binary
code of size $n \times b$, as we actually use it. We note that the
bit by bit addition of two codewords of size $n \times b$ is also
a codeword and hence the code is a binary linear code.

In these new codes, which will be constructed, we won't need the
redundancy bit of the third subset and hence the overall
redundancy will be reduced by one. The idea is to use the
technique of subsection~\ref{sec:limited_one}. The parity-check
matrix $H_1$, in subsection~\ref{sec:limited_one}, was used to
find the pattern of the error, and the parity-check matrix $H_2$
was used to find the location of a burst given its pattern up to a
cyclic shift.

This technique can be simply generalized for two-dimensional and
multidimensional codes. We will describe it only for
two-dimensional codes. Each of our constructions for
two-dimensional codes uses two components codes $\cC_1$ and
$\cC_2$. Assume that $\cC_1$ and $\cC_2$ are binary codes with
$r_1 \times n_1$ parity-check matrix $H_1$ and $r_2 \times n_2$
parity-check matrix $H_2$, respectively. We construct a
two-dimensional parity-check matrix $\cH$ for our two-dimensional
code as follows. $\cH$ is a two-dimensional matrix whose shape is
the shape of the two-dimensional codeword. In each position
$\Upsilon$ in this two-dimensional shape, $\cH$ has a column
vector of length $r_1+r_2$ which is a concatenation of a column
from $H_1$ and a column from $H_2$. The column from $H_1$ is
$\xi_1$ if $\xi_1$ is the color given to position $\Upsilon$ by
the first coloring (as mentioned before all our constructions can
be represented by the coloring method). Similarly, the column from
$H_2$ is $\xi_2$ if $\xi_2$ is the color given to position
$\Upsilon$ by the second coloring (to avoid confusion, in the
coloring method all positions are assigned with a color, by each
coloring, including the redundancy bits). Now, we can use a proof
similar to the proof of Lemma~\ref{lem:limit_burst} to show that
$\cH$ is a parity-check matrix for the required code. A
generalization for multidimensional codes is straightforward.

Finally, note that the same method can be also applied to the
three constructions presented in~\cite{BBZS}. Hence, we can supply
a parity-check matrix for each code constructed by these three
constructions.

\section{Conclusion and Open Problems}
\label{sec:conclusion}

As we wrote in the abstract, the main results of the paper are
summarized as follows:

\begin{enumerate}
\item A construction of small redundancy multidimensional codes
capable to correct a box-error. These codes and the box-error have
considerably more flexible parameters from previously known
constructions.

\item A novel method based on $D$ colorings of the $D$-dimensional
space for constructing $D$-dimensional codes correcting a
$D$-dimensional cluster-error of various shapes.

\item A transformation of the $D$-dimensional space into another
$D$-dimensional space in a way that a $D$-dimensional Lee sphere
is transformed into a shape located in a $D$-dimensional box of a
relatively small size. This transformation enables us to use the
previous constructions to correct a $D$-dimensional error whose
shape is a $D$-dimensional Lee sphere.

\item Applying the coloring method to correct more efficiently a
two-dimensional error whose shape is a Lee sphere.

\item A construction of one-dimensional and multidimensional codes
capable to correct a burst-error of length $b$ in which the number
of erroneous positions is $t$.

\item Applying the construction for correction of a Lee sphere
error and the construction for correction of a cluster-error with
small number of erroneous positions, to correct a $D$-dimensional
arbitrary cluster-error.
\end{enumerate}

All the codes we have constructed are binary. We didn't discuss
cluster-correcting codes over GF($q$), but most of our results can
be generalized straightforward for codes over GF($q$). We have
omitted some tedious proofs. The interested reader is referred
to~\cite{Yaa07} to see some of these proofs.

Clearly, our constructions do not cover all possible parameters.
Moreover, the redundancy of our codes is close to optimal, but not
optimal, so there is lot of ground for possible improvements with
possibly new construction methods. In fact, the main disadvantage
of our methods is that, for large $b$, the lengths of the known
$b$-burst-correcting codes and $b$-burst-locator codes are very
large.

Another question we didn't discuss in this paper is constructions
for cluster-correcting codes which correct a small cluster. This
question was considered in~\cite{ScEt}. The construction of
$D$-dimensional cluster-correcting code which corrects a
$D$-dimensional Lee sphere error with radius one is important in
this connection. As we mentioned in subsection~\ref{sec:multi_Lee}
we don't know the maximum number of $(2D+1)$-burst-locator codes,
among the $D$ component codes, that we can use.

The next question is how to implement the coloring method for
correction of multidimensional Lee sphere errors with dimension
greater than two and radius greater than one?

Finally, we still don't know and even don't have any indication
what should be the excess redundancy of an optimum code which
corrects an arbitrary multidimensional cluster-error. Even the
two-dimensional case is far from being resolved. The gap between
the lower and upper bounds of Theorem~\ref{thm:lower_excess} and
Theorem~\ref{thm:upper_excess}, respectively, is quite large and
we believe that both bounds can be improved.

\section*{Acknowledgment}

The authors wish to thank Khaled Abdel-Ghaffar for
providing~\cite{Abd86}. They also thank Khaled Abdel-Ghaffar and
Jack Wolf for helpful discussions.


%
%
%
%
%
%

\end{document}